\documentclass[12pt]{amsart}
\usepackage[utf8]{inputenc}
\usepackage[english]{babel}
\usepackage{amsmath,amssymb,amsfonts,amsthm}
\usepackage{amscd} 
\usepackage{mathtools} 
\usepackage{longtable}
\usepackage[noadjust]{cite} 
\usepackage{enumitem}
\usepackage{float}
\usepackage[mathcal]{euscript} 
\usepackage[unicode]{hyperref}
\usepackage{ragged2e}
\usepackage{xfrac} 

\usepackage{graphicx, array}
\graphicspath{{.}{figures/}}
\usepackage[labelfont=small,labelformat=simple]{subcaption}

\usepackage{xcolor}

\theoremstyle{plain}
\newtheorem{theorem}{Theorem}[section]
\newtheorem{lemma}[theorem]{Lemma}

\theoremstyle{definition}
\newtheorem{definition}[theorem]{Definition}

\newtheorem{remark}[theorem]{Remark}

\newtheorem{problem}[theorem]{Problem}

\begin{document}

\renewcommand{\datename}{}

\renewcommand{\abstractname}{Abstract}
\renewcommand{\refname}{References}
\renewcommand{\tablename}{Table}
\renewcommand{\figurename}{Figure}
\renewcommand{\proofname}{Proof}

\title[Recent advances in the monodromy theory]{Recent advances in the monodromy theory of 
integrable Hamiltonian systems}

\author{N. Martynchuk$^{1}$}
\thanks{
$^1$ Department Mathematik, FA University Erlangen-N{\"u}rnberg, Cauerstr. 11, D-91058 Erlangen, Germany \\
$^2$ Bernoulli Institute, University of Groningen, P.O. Box 407, 9700 AK Groningen, The Netherlands \\
$^3$ Duke Kunshan University, No. 8 Duke Avenue, Kunshan, Jiangsu Province, China 215316 \\
E-mail addresses: \textit{martynchuk@math.fau.de, h.w.broer@rug.nl, \\ k.efstathiou@dukekunshan.edu.cn}}

\author{H.~W. Broer$^2$}

\author{K. Efstathiou$^3$}

  \begin{abstract}
The notion of monodromy was introduced by J.~J. Duistermaat as the first obstruction to the existence of global action coordinates in integrable Hamiltonian systems. This invariant was extensively studied since then and was shown to be non-trivial in various concrete examples of finite-dimensional integrable systems.
  The goal of the present paper is to give a brief overview of monodromy and discuss some of its generalisations. In particular, we will discuss the monodromy around a focus-focus singularity and the notions of quantum, fractional and scattering monodromy. The exposition will be complemented with a number of examples and open problems.

\hspace{-5.7mm} \textit{Keywords}: Action-angle coordinates; Hamiltonian system; Liouville integrability; Monodromy; Quantisation.
 
\end{abstract}

\maketitle

\newpage

\section{Introduction} \label{sec/introduction}

In the context of finite-dimensional integrable Hamiltonian systems, the notion of monodromy was introduced by Duistermaat in his seminal paper \cite{Duistermaat1980} published in 1980. He defined his notion of monodromy as the (usual) monodromy of a certain covering map that can naturally be defined for a given integrable system. To be more specific, assume that we are given $n$ independent functions in involution
$(F_1, \ldots, F_n)$ on a symplectic manifold $M$ of real dimension $2n$ 
\footnote{We recall that an integrable Hamiltonian system on a symplectic $2n$-manifold $M$ is specified by $n$ independent functions in involution $F_1, \ldots, F_n$. Typically, $F_1 = H$ is the Hamiltonian of the system and $F_2, \ldots, F_n$ are additional first integrals.}. These functions give rise to the 
so-called \textit{integral} or the \textit{momentum map}
$$F = (F_1, \ldots, F_n) \colon M \to \mathbb R^n$$ and the (defined on an open subset $U \subset \mathbb R^n \times M$)  action
$$
G \colon U \subset \mathbb R^n \times M \to M, \ G(t_1, \ldots, t_n)(x) = g_1^{t_1} \ldots g_n^{t_n}(x),
$$
where $g^t_i$ is the Hamiltonian flow associated to $F_i$. Observe that the action $G$ leaves the fibers
$F^{-1}(f) \subset M$ of $F$ invariant since the functions $F_1, \ldots, F_n$ are in involution.

For simplicity, we shall for the moment consider the case when all of the fibers $F^{-1}(f)$
are compact and connected. Then the action $G$ is a global $\mathbb R^n$ action on $M$. Moreover,
for each regular value $f$ in the image of $F$, the isotropy group $G_f$ is an $n$-dimensional lattice $\mathbb Z^n \subset \mathbb R^n.$ In particular, regular fibers $F^{-1}(f)$ are $n$-dimensional tori; see Arnol'd-Liouville theorem \cite{Liouville1855, Arnold1968, Arnold1978} for detail.
The collection of the lattices $G_f$, with $f$ in the set $R \subset \textup{image}(F)$ of the regular values  of $F$, is a subset of $\mathbb R^n \times R.$ The natural projection 
$\Pr \colon \mathbb R^n \times R \to R$ gives rise to the covering map
\begin{equation} \label{eq/covering}
\Pr \colon \bigcup_{f\in R} G_f \to R.
\end{equation}
This is the covering that we mentioned above. In the paper \cite{Duistermaat1980}, the monodromy of the torus fibration 
$F \colon F^{-1}(R) \to R$ was defined as the (usual) 
monodromy of the covering \eqref{eq/covering}, that is, the representation of the fundamental group $\pi_1(R,f_0)$ in the group
 of automorphisms of $G_{f_0} \equiv \mathbb Z^n$ (the representation is given by lifting paths from $\pi_1(R,f_0)$ to the total space
of the covering \eqref{eq/covering}).

We note that Duistermaat's original definition included the case of Lagrangian torus fibrations over an arbitrary 
manifold (not necessarily an open subset of $\mathbb R^n$). We will not pursue this generality here.

Since Duistermaat's work \cite{Duistermaat1980}, non-trivial monodromy was found in various concrete integrable systems of physics and classical mechanics. The first such example is the spherical pendulum, which is an integrable system
that describes the motion of a particle on the unit sphere in $\mathbb R^3$ in the linear gravitational potential\footnote{
For this system, the functions $F_1 = H$ and $F_2 = J$ are the restrictions of the functions
$
H = \frac{1}{2}\|p\|^2+ q_3$ and $
J = q_1p_2 - q_2 p_1,
$
defined on $T^{*}\mathbb R^3,$ to $T^{*}S^2 \subset T^{*}\mathbb R^3$.}.
The monodromy of the spherical pendulum was
observed to be non-trivial by R. Cushman and computed by J.~J. Duistermaat in the same paper \cite{Duistermaat1980}.
It turned out that $\pi_1(R,f_0)$ is isomorphic to $\mathbb Z$ in this case (see Fig.~\ref{fig/BD_spherical_pendulum}) and that the monodromy is given by the matrix
\begin{equation} \label{eq/monodromy_matrix}
 M_\gamma = \begin{pmatrix} 1 & 1 \\ 0 & 1\end{pmatrix}.
 \end{equation}
Here $\gamma$ corresponds to the generator of the group $\pi_1(R,f_0) \equiv \mathbb Z.$
We shall return to this example and to the computation of the monodromy matrix later in this paper.

\begin{figure}[ht]
\begin{center}
\includegraphics[width=0.95\linewidth]{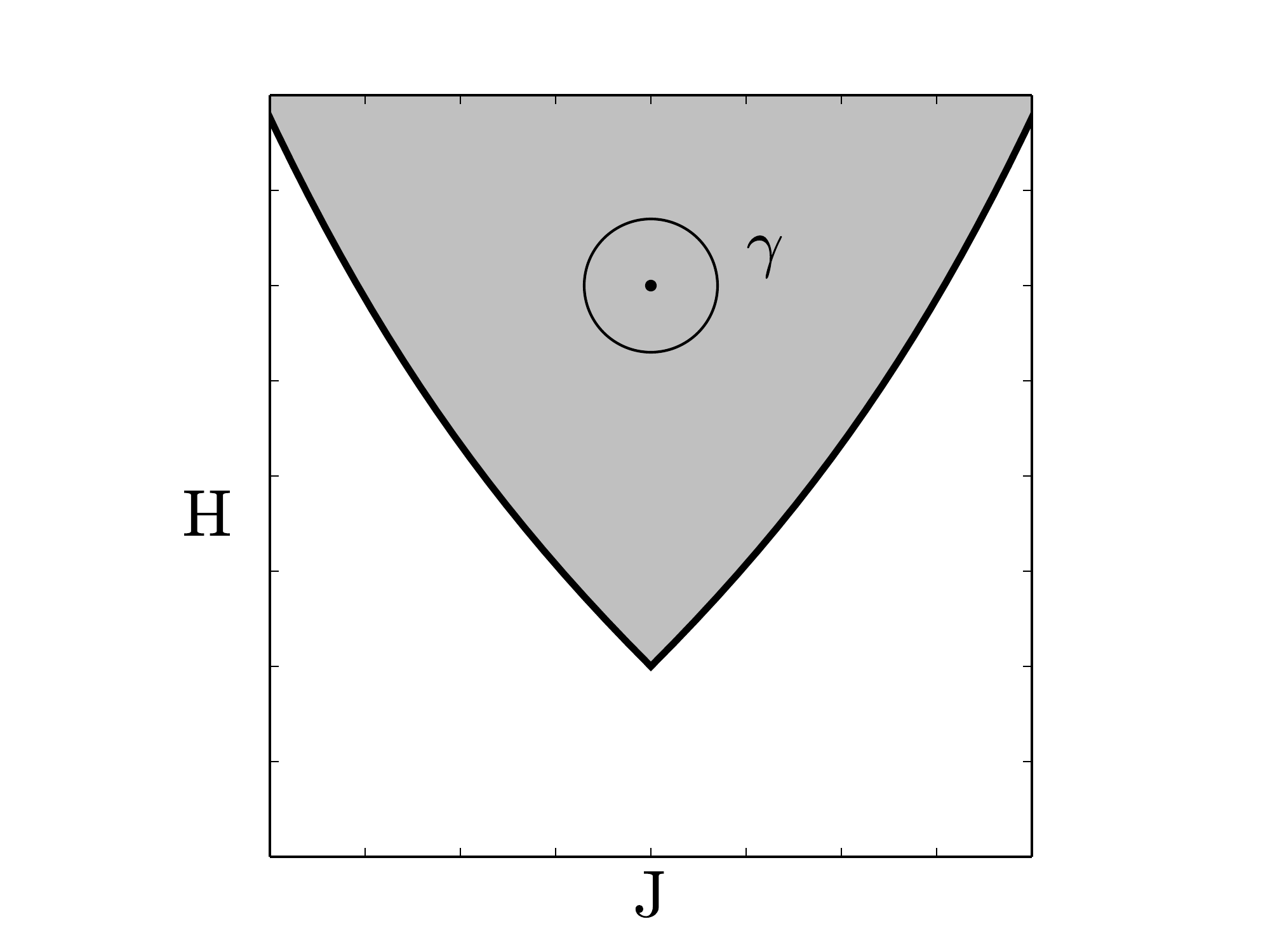}
\caption{Bifurcation diagram for the spherical pendulum and the generator $\gamma$ of  $\pi_1(R)$. }
\label{fig/BD_spherical_pendulum}
\end{center}
\end{figure}

Another example, which is probably the simplest one, is the so-called champagne bottle system (a particle in a Mexican hat potential). 
For this system, the monodromy was computed by L. Bates in \cite{Bates1991}. It turns out that also in this case, the fundamental group $\pi_1(R,f_0)$ is  isomorphic to $\mathbb Z$ and the corresponding monodromy matrix  is given by Eq.~\eqref{eq/monodromy_matrix}. 
 
Several other examples of integrable Hamiltonian systems with non-trivial monodromy are 
the quadratic spherical pendulum
\cite{Duistermaat1980, Bates1993, Efstathiou2005, Cushman2015}, the
coupled angular momenta \cite{Sadovskii1999}, the Lagrange top \cite{Cushman1985}, 
the Hamiltonian Hopf bifurcation \cite{Duistermaat1998},  the Jaynes-Cummings model  \cite{JaynesCummings1963, Pelayo2012, DullinPelayo2016}, the hydrogen atom in crossed fields \cite{Cushman2000}, and the Euler two-center problem \cite{Waalkens2004, Martynchuk2018}.  We note that monodromy can naturally be generalised to integrable non-Hamiltonian systems \cite{Cushman2001, Zung2002}; see also \cite{Bolsinov2015}
for a discussion on monodromy in the context of the Hamiltonisation problem. This invariant can also be extended to the setting of near-integrable  systems \cite{Rink2004, Broer2007a, Broer2007}, which is relevant for applications since real physical systems are seldom integrable. 

It was later understood that most of the known examples of integrable  systems with non-trivial monodromy have one common property, namely, the existence of the so-called
\textit{focus-focus points}. For instance, in the case of the spherical pendulum, this is the unstable equilibrium when the pendulum is at the top of the sphere. In the case of the Mexican hat potential, this is the unstable equilibrium 
when the particle is on the `top of the hat'. The precise result, which is sometimes referred to as the \textit{geometric
monodromy theorem}, was obtained first by L.~M. Lerman and Ya.~L. Umanski{\'i}  \cite{Lerman1994} in the case of a single 
focus-focus point and later by V.~S. Matveev \cite{Matveev1996} and N.~T. Zung \cite{Zung1997} in the case of arbitrary many focus-focus points
on a singular focus-focus fiber. We note that outside the context of integrable Hamiltonian system,  this result was already
obtained by Y. Matsumoto in \cite{Matsumoto1989}. We also note that in the context of complex geometry, the 
 geometric monodromy theorem
follows from the Picard-Lefschetz theory; see \cite{Zung1997, Audin2002, Bolsinov2004} for details. We shall come back to 
case of focus-focus singularities later in this paper, in connection with the \textit{classical} Morse theory 
and principal circle bundles; this is the content of the recent topological theory of monodromy developed in \cite{Martynchuk2018}.

Another breakthrough in the monodromy theory was the quantum formulation of this invariant; first, for the quantum spherical pendulum \cite{Cushman1988, Guillemin1989}
and later, in more generality, by S. V{\~u} Ng{\d{o}}c \cite{Vu-Ngoc1999}. The main idea is that in a quantum integrable system, the joint spectrum of the commuting operators locally has the form of a lattice. Globally, this does not have to be the case, and one can observe a lattice defect in the joint spectrum when transporting an elementary cell around a singularity; see Fig.~\ref{fig/BD_spectrum_pendulum}. 
\begin{figure}[htbp]
  \includegraphics[width=\linewidth]{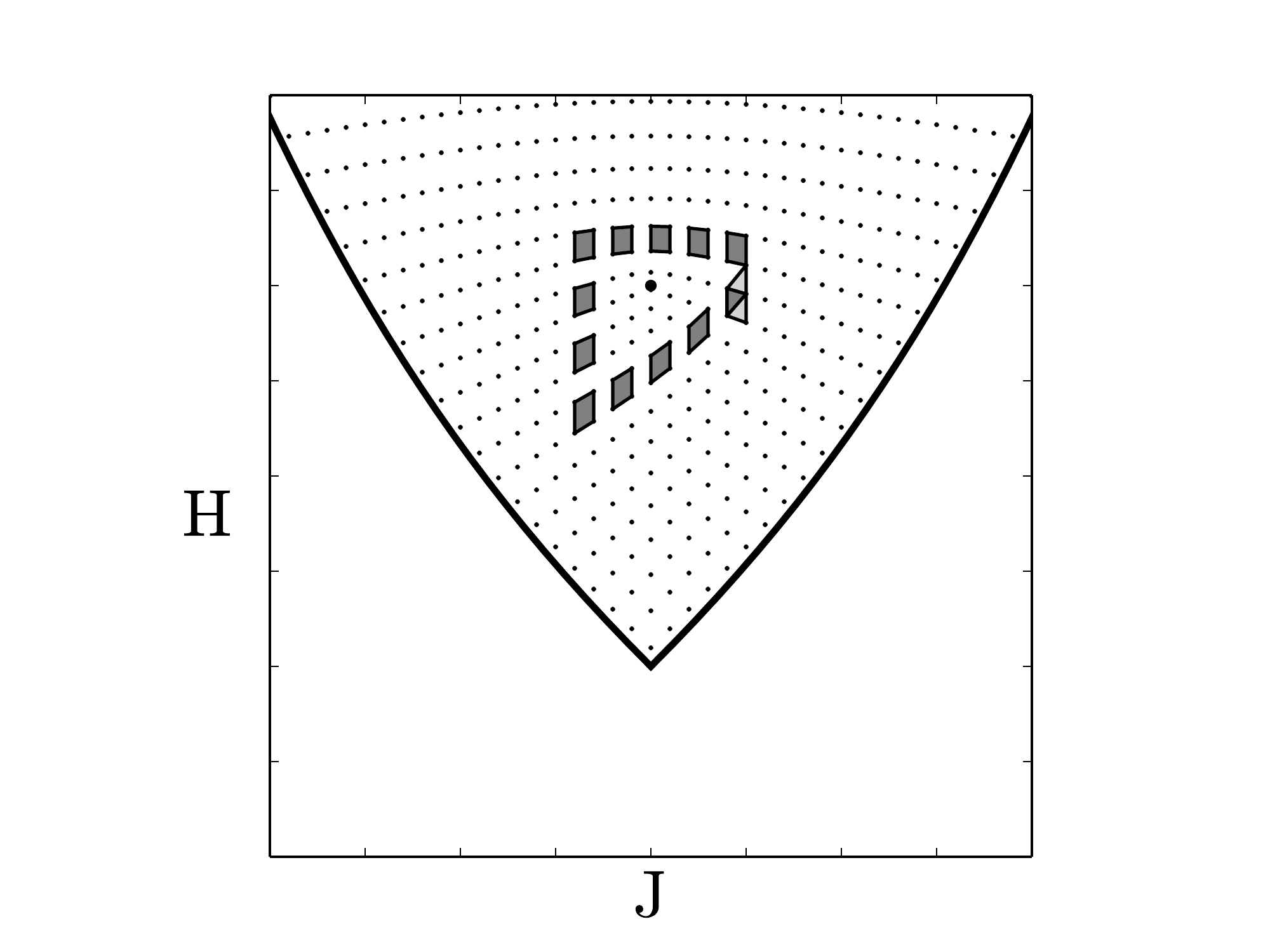}
  \caption{The joint spectrum of the quantum spherical pendulum ($\hbar = 0.1$), and the transport of an elementary cell around the focus-focus point.}
  \label{fig/BD_spectrum_pendulum}
\end{figure}
This lattice 
defect is usually interpreted as the non-existence of smooth global quantum number assignment for the given quantum integrable system. We note that this is very similar to what happens classically when one looks at the action coordinates and the so-called integer affine structure \cite{Zung1997}. We also note that quantum monodromy is always given by the classical monodromy of the underlying classical integrable Hamiltonian system \cite{Vu-Ngoc1999}.

This is, in short, what is classically known about monodromy. More recently, several generalised versions of monodromy have been defined. The most important and general of these are the so-called
 fractional and scattering monodromies as well as their quantum analogues. 
 The notion of fractional monodromy was introduced in the paper \cite{Nekhoroshev2006} as a generalization of the usual Duistermaat's monodromy 
 (sometimes referred to as Hamiltonian monodromy) 
 to the case of singular fibrations; it naturally appears in integrable systems with hyperbolic singularities.
 Scattering monodromy appears in completely integrable systems with non-compact invariant manifolds; it was originally defined by L. Bates and R. Cushman in \cite{Bates2007} for a two degree of freedom hyperbolic oscillator and later generalized in the works \cite{DullinWaalkens2008, Efstathiou2017} and \cite{Martynchuk2019}. 
 
 The main goal of the present paper is to give a concise and systematic overview of the monodromy theory, and of some of the recent developments in this field. Our main focus will be on the classical notion of monodromy and some of the generalised versions of this invariant. We will complement our exposition with various concrete examples and formulate a few open problems. For a more thorough exposition of the state of the art of the monodromy theory and integrable systems, we refer the reader
 to \cite{Bolsinov2004, Bolsinov2006, Cushman2015, Zhilinskii2010, Sepe2018, Martynchuk2018}. Several parts of this work appeared in a more extended form in \cite{Martynchuk2018}.

\section{Preliminaries on Hamiltonian monodromy}

The notion of Hamiltonian monodromy\footnote{Duistermaat's notion of monodromy is usually referred to as Hamiltonian monodromy to distinguish it from other types of monodromy, such as fractional monodromy or monodromy of a covering map.} was originally introduced as the first obstruction to the existence of global action angle-coordinates in integrable systems \cite{Duistermaat1980}. We briefly review a construction of these coordinates here and explain the relation to the definition of Hamiltonian monodromy given in the Introduction. Then we discuss a connection of Hamiltonian monodromy to Picard-Lefschetz theory, the latter being a very classical situation in which monodromy of non-singular hypersurfaces appear. The discussion  continues in the next section, where we review
the classical theorem which describes the monodromy around a focus-focus singularity and discuss several more recent results.

\subsection{Liouville integrability, action-angle coordinates and monodromy}

We recall that a Hamiltonian system 
\begin{equation*}
\dot{x} = X_H, \ \ \omega(X_H, \cdot) = - dH,
\end{equation*}
on a $2n$-dimensional symplectic manifold $(M, \omega)$ is called Liouville integrable if there exist almost everywhere independent functions
$F_1 = H, \ldots, F_n$ that are in \textit{involution} with respect to the symplectic form $\omega$:
$$
\{F_i,F_j\} = \omega(X_{F_i}, X_{F_j}) = 0.
$$
We note that by definition, for each $i$ and $j$, the function $F_i$ is invariant with respect to the Hamiltonian flow of $F_j$; in particular,  the functions $F_i$ are first integrals of the flow of $X_H$.
Various Hamiltonian systems, such as the Kepler problem, the spherical pendulum, the geodesic flow on an ellipsoid,  Euler, Lagrange and Kovalevskaya tops, the Calogero-Moser systems,
are integrable in this sense. 

The map $F = (F_1, \ldots, F_n)$ consisting of the integrals $F_i$ is called the \textit{integral map} (or the \textit{energy-momentum map}) of the integrable system. It encodes both the dynamics ($F_1 = H$) and the symmetry associated to the system. A central problem in the theory of integrable systems is to understand the geometry of such integral maps; in other words, to classify them up to a topological, smooth or symplectic equivalence.

It is well-known that, in the case when the function $F$ is proper, any regular fiber $F^{-1}(\xi_0)$ is an $n$-dimensional torus (or a union of several $n$-tori). Moreover, a small tubular neighborhood of any such torus is
a trivial torus bundle $D^n \times T^n$ admitting \textit{action-angle coordinates} 
$$I \in D^n \ \mbox{ and } \ \varphi \mod 2\pi \in T^n, \ \ \omega = dI \wedge d\varphi.$$ 
This is the content of the Arnol'd-Liouville theorem \cite{Liouville1855, Arnold1968, Arnold1978}.
It follows from the existence of action-angle coordinates that the motion (that is, the flow of $X_H$) is quasi-periodic on each torus $\{\xi\} \times T^n$.

The above coordinates are sometimes referred to as semi-local since they exist in a neighborhood of a given invariant torus. The global situation (of when do such coordinates exist globally)
was clarified by Nekhoroshev \cite{Nekhoroshev1972} and  Duistermaat \cite{Duistermaat1980}. We briefly review  a few main results of these works below. 

Let 
$R \subset image(F)$ denote the set of the regular values of $F$ that are in the image of $F$. Assume for the moment that all of the fibers $F^{-1}(f)$ are compact and connected. Then global action-angle
coordinates exist if the following two conditions are satisfied (see \cite{Nekhoroshev1972}):
$$
\pi_1(R, f_0) = 0 \ \mbox{ and} \ H^2(R, \mathbb R) = 0.
$$
Otherwise, the torus bundle $F \colon F^{-1}(R) \to R$ is not necessarily globally trivial, and certain obstructions to the triviality of this bundle
appear; see \cite{Duistermaat1980}. One of such obstructions
is monodromy, which we have briefly discussed in the introduction. It is an obstruction in the sense that its non-triviality entails to the non-existence of global action coordinates. To see this, let us assume for simplicity that the symplectic form $\omega$ is exact:
$\omega = d\eta$. 
Then the action coordinates $I = (I_1, \ldots, I_n)$ can be defined by the formula
$$
I_i = \dfrac{1}{2\pi}\int I_i d\varphi_i = \dfrac{1}{2\pi}\int_{\alpha_i} I d\varphi =  \dfrac{1}{2\pi}\int_{\alpha_i} \eta + c_i,
$$
where $\alpha_i$ is the $\varphi_i$-cycle on the corresponding Liouville torus $F^{-1}(f)$ and $c_i$ does not depend on $f$. The cycles $\alpha_1, \ldots, \alpha_n$ form a basis of the first \textit{integer} homology group of $F^{-1}(f)$. But this homology group can be identified with the isotropy group $G_f$ of the global $\mathbb R^n$ action on $F^{-1}(f)$; see Introduction (Sec.~\ref{sec/introduction}).  Thus, the non-triviality of monodromy of the covering, Eq.~\eqref{eq/covering}, formed by the lattices
$G_f$ implies that it is not possible to choose the cycles $\alpha_1, \ldots, \alpha_n$ in a continuous way over 
$R$: transports of these homology cycles along different paths do not give the same result. In particular, it is not possible to choose the action coordinates in a globally \textit{smooth} way: transports along different paths result in different sets 
of action coordinates $I$ and $I'$ related by a transformation $I = MI'$, where $M \in \textup{SL}(n,\mathbb Z)$. After excursions along elements of $\pi_1(R,f_0)$, we get the monodromy automorphisms, described in the Introduction.

\subsection{Picard-Lefschetz theory} \label{sec/PLtheory} 

In the context of fibrations by complex tori, the notion of Hamiltonian monodromy is essentially the classical 
monodromy that appears in Picard-Lefschetz theory.

Let $\mathbb C^2$ be the complex two-plane with complex coordinates $(z,w)$. Following \cite{Bolsinov2004}, consider the symplectic transformation 
\begin{equation} \label{eq/A}
A(z,w) \to (w^{-1},zw^2)
\end{equation} 
(defined for $w\ne 0$). Let the compact manifold $M$ be defined by gluing the boundary solid tori of
\begin{equation} \label{eq/u1}
U_1 = \{ (z,w) \in \mathbb C^2 \mid |zw| \le \varepsilon, |z| \le 1, |w| \le 1 \}
\end{equation}
using this transformation. (The boundary solid tori of $U_1$ are given by the sets $\{(z,w) \in U_1\mid |z| = 1\}$ and $\{(z,w) \in U_1\mid |w| = 1\}$.) Observe that the function $f \colon \mathbb C^2 \to \mathbb C$ defined by
$$
f(z,w) = zw
$$ 
descends to a smooth function on this manifold. It has one critical fiber: the preimage of the origin in $\mathbb C$. All of the other fibers are regular two-tori. Let $\gamma$ be a small circle in $\mathbb C$ around the origin. According to the Picard-Lefschetz formula \cite{Arnold2012}, the monodromy of $f$ along $\gamma$ is given by the matrix
\begin{equation*}
 M_\gamma = \begin{pmatrix} 1 & 1 \\ 0 & 1\end{pmatrix}.
 \end{equation*}
 
Now observe that the holomorphic function $f$ can be viewed as an energy-momentum map of a real integrable Hamiltonian system on $M$: the functions in involution are given by the real and imaginary part of the function $f$; see \cite{Flaschka1988}.
By a topological definition of Hamiltonian monodromy in terms of homology cycles, this matrix is the monodromy matrix along $\gamma$
associated to this integrable system.

For the above argument, it is important that the phase space is a complex manifold and that $f$ is a holomorphic  (meromorphic)
function on this manifold. We note that in a general situation, an integrable Hamiltonian system is only defined on a real symplectic manifold and, even if the manifold can be endowed with a complex structure, the integrals of motion are not always
meromorphic functions. Therefore, the Picard-Lefschetz formula is not always applicable; at least, not directly. Nonetheless, in various examples of integrable systems the integrals of motions are polynomials and it is possible to complexify them. Then one can 
use the Picard-Lefschetz theory in the complexified domain and deduce information about monodromy in the original system. We refer to
\cite{Audin2002, Beukers2003, Sugny2008} for more information.

\section{Hamiltonian monodromy} \label{section/Hamiltonian_monodromy}

In this section, we continue our discussion of Hamiltonian monodromy. We review the  \textit{geometric monodromy theorem}, which describes the monodromy around a focus-focus singularity. This central result in monodromy theory allows one to compute monodromy in various concrete integrable systems by computing the complexity of the focus-focus fibers of such systems. We then explain a dynamical manifestation of non-trivial Hamiltonian monodromy. Afterwards, we come back to the spherical pendulum and discuss the monodromy from a different point of view based on Morse theory and
Chern numbers (a general situation is treated in the work \cite{Martynchuk2020}). We conclude this section with an extension of Hamiltonian monodromy to nearly integrable systems.

\subsection{Monodromy around a focus-focus singularity}

Hamiltonian monodromy was first observed to be non-trivial in concrete integrable systems of classical mechanics and molecular physics. It was later observed  that in the typical case of $n = 2$ degrees of freedom, non-trivial monodromy is manifested by the presence of the 
so-called \textit{focus-focus} points of the integral fibration $F$; see \cite{Lerman1994, Matveev1996, Zung1997}. 
(The singular point $z = w = 0$ of the function $f = zw$ from Subsection~\ref{sec/PLtheory}
is an example  of a focus-focus point.)
Such a result is often referred to as \textit{geometric monodromy theorem}. Below we discuss a few different approaches to this theorem. 

First, let us recall the notion of a focus-focus singularity.

\begin{definition}
Consider a two-degree of freedom integrable system $F = (H,J) \colon M \to \mathbb R^2$ on a $4$-manifold $M$.
Let $x_0$ be a rank zero singular point of $F$, that is,   $dF_{x_0} = 0$. The point $x_0$ is called a focus-focus point of 
$F = (H,J)$ if the Hessians $d^2_{x_0}H$ and $d^2_{x_0}J$ are independent and there exists local canonical coordinates near $x_0$ such that 
\begin{align*}
d^2_{x_0}H &= A_1(dp_1dq_1+dp_2dq_2) + B_1(dp_1dq_2-dp_2dq_1) \\ 
d^2_{x_0}J &= A_2(dp_1dq_1+dp_2dq_2) + B_2(dp_1dq_2-dp_2dq_1).
\end{align*}
\end{definition}

\begin{remark}
The focus-focus singularity is an example of a non-degenerate singularity of an integrable system.
Alongside focus-focus points, there are also other types of non-degenerate singular points of integrable two-degrees of freedom systems: elliptic-elliptic, hyperbolic-hyperbolic, elliptic-regular, etc.; see \cite{Bolsinov2004} for details. 
\end{remark}

\begin{remark}
We note that by the Williamson theorem, not only the quadratic parts of $H$ and $J$, but also the map $F = (H,J)$ itself can be put into a normal form near a singular focus-focus point: there exist local canonical coordinates near $x_0$ such that 
\begin{align*}
H &= H(p_1q_1+p_2q_2, p_1q_2-p_2q_1) \\ 
J &= J(p_1q_1+p_2q_2, p_1q_2-p_2q_1).
\end{align*}
We note that a similar statement holds for other types of non-degenerate singular points; see \cite{Bolsinov2004}.
\end{remark}

Assume that we are given a proper integral map $F$ with an isolated critical value $f_0$ such that the critical fiber $F^{-1}(f_0)$ contains
a (finite) number of focus-focus points.
The geometric monodromy theorem describes the monodromy of $F$ around $f_0$ in this situation in terms of the number
 of the focus-focus points.

\begin{theorem} \textup{(Geometric monodromy theorem, \cite{Matsumoto1989, Lerman1994, Matveev1996, Zung1997, Vu-Ngoc1999})} \label{theorem/geometric_monodromy_theorem0}
Monodromy around a focus-focus singularity is given
 by the matrix
 \begin{equation*}
 M = \begin{pmatrix} 1 & m \\ 0 & 1\end{pmatrix},
 \end{equation*}
 where $m$ is the number of the focus-focus points on the singular fiber.
\end{theorem}

One way to prove this theorem is to prove that the number $m$ of the focus-focus point on a singular focus-focus fiber $F^{-1}(f_0)$
(also called the \textit{complexity} of this fiber)
is a complete topological invariant of the Liouville fibration in a tubular neighborhood of this fiber $F^{-1}(f_0)$; see \cite{Matveev1996, Zung1997}. The monodromy is a particular 
invariant of this fibration, and is thus a function of the number $m$ of the focus-focus points. To prove the geometric monodromy theorem, it is  sufficient to prove the statement for a particular example of an integrable system with $m$ 
focus-focus points. The rest follows from Picard-Lefschetz theory; cf. Subsection~\ref{sec/PLtheory}. We refer to \cite{Zung1997} for details.

\begin{remark}
We have noted above that the complexity is a complete semi-local topological invariant of a focus-focus singularity; see 
\cite{Zung1997, Bolsinov2004}. This is not the case symplectically:
there exist  infinitely many (semi-locally) non-symplectomorphic Lagrangian fibrations even in the case of complexity $m = 1$; see \cite{Vu-Ngoc2003}. We note that a similar result does not hold even in the smooth category: there exist smoothly non-equivalent Lagrangian fibrations in the case of $m \ge 2$ focus-focus points on a given focus-focus fiber; see the works \cite{Bolsinov2004, Izosimov2011, Bolsinov2019} for details.
\end{remark}

\begin{remark}
We note that in concrete problems of physics and classical mechanics, the complexity of focus-focus fibers is usually small. This can be proven rigorously in many cases in terms of the topology of the underlying symplectic manifold \cite{Smirnov2013}. For instance, in
$\mathbb R^4$ one can only have  complexity $m = 1$ focus-focus fibers ($\mathbb R^4$ does not contain Lagrangian spheres \cite{Bolsinov2006}).
For integrable systems 
on $T^{*}S^2$, one can have complexity $m = 1$ or $m = 2$, but not $3$ or more. We refer to the work 
\cite{Smirnov2013} for details.
\end{remark}

 A related  result in the context of the focus-focus singularities is that they come  with a Hamiltonian circle action \cite{Zung1997, Zung2002}.

\begin{theorem}\textup{(Circle action near focus-focus, \cite{Zung1997, Zung2002})} \label{zungaction1} 
In a neighborhood of a singular focus-focus fiber, there exists a unique (up to orientation) Hamiltonian 
circle action which is free outside the singular focus-focus points.
Near each focus-focus point, the momentum of the circle action can be written as
\begin{align*}
  J = \frac12 (q_1^2+p_1^2) - \frac12 (q_2^2+p_2^2)
\end{align*}
for some local canonical coordinates $(q_1, p_1, q_2, p_2)$. In particular, the circle action defines the anti-Hopf fibration near each singular point.
\end{theorem}

 One implication of Theorem~\ref{zungaction1} is that it allows one to give a different proof of the geometric monodromy theorem by looking at the circle action. For example, one can apply the Duistermaat-Heckman theorem; see
 \cite{Zung2002}.  A related and purely topological proof will be given below on the example of the spherical pendulum, following the point of view of \cite{Martynchuk2018, EfstathiouMartynchuk2017, Martynchuk2017, Martynchuk2020}.
 For other approaches to the geometric monodromy theorem, we refer the reader to \cite{Vu-Ngoc2000, Audin2002, Cushman2015, Efstathiou2017}.
 
 \subsection{Dynamical manifestation of monodromy}
 
 In this subsection we briefly comment on implications of non-trivial monodromy for dynamics. More specifically, we 
 make a connection to the so-called \textit{rotation number}  \cite{Cushman2015}.
 
 We assume that the energy-momentum map $F = (H,J)$ is such that all of the fibers $F^{-1}(f)$ are compact and connected. Moreover, 
we assume that $F$ is invariant under the Hamiltonian circle action given by the Hamiltonian flow $\varphi^t_J$  of $J$. 
Let $F^{-1}(f)$ be a regular torus. Consider a point $x \in F^{-1}(f)$ and the orbit of the circle action passing through this point. 
The trajectory $\varphi^t_H(x)$ leaves the orbit of the circle action at $t = 0$ and then returns back to the same orbit at some time $T > 0$. The time $T$ is called the \textit{the first return time}.
The \textit{rotation number} $\Theta = \Theta(f)$ is defined by $\varphi^{2\pi {\Theta}}_J(x) = \varphi^{T}_H(x)$.  With this notation, there is the following result.

\begin{theorem} \textup{(Monodromy and rotation number, \cite{Cushman2015})} \label{theorem/rotation_number_mon}
 The Hamiltonian monodromy of the torus bundle $F \colon F^{-1}(\gamma) \to \gamma$ is given by
$$
\begin{pmatrix}
 1 & m \\
 0 & 1
\end{pmatrix} \in \mathrm{SL}(2,\mathbb Z),
$$
where $-m$ is the variation of the rotation number $\Theta$.
\end{theorem} 

We note that this theorem can be used as a powerful analytic  tool for the computation of monodromy in specific examples of integrable systems with a circle action. We refer
to \cite{Vu-Ngoc2000, Cushman2015, Efstathiou2017} for details. For another dynamical manifestation of monodromy,
see \cite{Delos2008}.

\subsection{The spherical pendulum}

We now come back to the case of the spherical pendulum and prove that the monodromy matrix of this system is given by Eq.~\ref{eq/monodromy_matrix}. We shall mainly focus on a topological idea which goes back to R. Cushman and F. Takens and which has been developed in the works
\cite{Martynchuk2020, Martynchuk2017, EfstathiouMartynchuk2017}.

We recall that the spherical pendulum is a mechanical Hamiltonian system that
describes the motion of a particle moving on the sphere 
\begin{align*}
S^2 = \{ (x,y,z) \in \mathbb R^3 \colon x^2 + y^2 + z^2 = 1 \}
\end{align*}
in the linear gravitational potential $V(x,y,z) = z.$
The phase space is $T^{*}S^2$ with the standard symplectic structure. The Hamiltonian is given by  
$$H = \frac{1}{2}(p_x^2+p_y^2+p_z^2) + V(x,y,z)$$ 
the total energy of the
pendulum. Since (the component of) the angular momentum $J = xp_y - yp_x$ is conserved,
the system is Liouville integrable.
The bifurcation diagram of the  energy-momentum map 
$$F = (H,J) \colon T^{*}S^2 \to \mathbb R^2,$$ that is, the set of the critical values of this map, is shown in Fig.~\ref{fig/BD_spherical_pendulum}. 

Consider the closed path $\gamma$ around the isolated critical value; see Fig.~\ref{fig/BD_spherical_pendulum}.  It was shown by Duistermaat in \cite{Duistermaat1980} using an analytic argument that the monodromy along $\gamma$ is given by the matrix 
\begin{equation} \label{eq/mon_p}
 M_\gamma = \begin{pmatrix} 1 & 1 \\ 0 & 1\end{pmatrix}.
 \end{equation}
 \begin{remark}
Duistermaat's proof is based on the computation of the action coordinates. To be more specific, observe that for the spherical pendulum, there are `natural' actions coming from a separation of the system in spherical coordinates. One of these actions is simply given by the function $J$; it is globally defined on the phase space  $T^{*}S^2$. The other one is an elliptic integral. One can deduce the monodromy from the derivatives of the second action when $J$ approaches zero; see \cite{Duistermaat1980} for details. We note that this kind of approach can be used more generally; it reduces the computation of monodromy to studying certain limits of elliptic integrals. 
\end{remark}

We note that the above result can  directly be obtained from the geometric monodromy theorem, Theorem~\ref{theorem/geometric_monodromy_theorem0}. Indeed, it can be shown that the isolated critical value is a focus-focus singularity of complexity $1$ (there is one and only one unstable equilibrium of the pendulum).

Below, following  the work \cite{Martynchuk2020}, we shall give a different proof of Eq.~\ref{eq/mon_p}, without computing the action coordinates or invoking the geometric monodromy theorem, but using only topological ideas.
 
The first step, is to observe that $J$ generates a Hamiltonian circle action on $T^{*}S^2$. It follows that any orbit of this action
on $F^{-1}(\gamma(0))$ can be transported along $\gamma$. Let $(a,b)$ be a basis of $H_1(F^{-1}(\gamma(0)))$, where $b$ is given by the homology class of such an orbit. Then the corresponding Hamiltonian
monodromy matrix along $\gamma$ is given by
\begin{equation*}
 M_\gamma = \begin{pmatrix} 1 & m_\gamma \\ 0 & 1\end{pmatrix}
 \end{equation*}
 for some integer $m_\gamma$. We now prove that the integer $m_\gamma \ne 0$; this argument is due to R. Cushman.

\begin{proof}
Observe that the points 
$$P_{min} = \{p = 0, z = -1\} \ \mbox{ and } \ P_{c} = \{p = 0, z = 1\}$$ are the only critical points of $H$, and they are non-degenerate. We have $H(P_{min}) = -1$ and $H(P_{c}) =1$. From the Morse lemma, for small $\varepsilon > 0$ ($\varepsilon$ should be less than $2$), the manifold $H^{-1}(1 - \varepsilon)$ is diffeomorphic to the $3$-sphere $S^3$. On the other hand,  it can be shown that $H^{-1}(1 + \varepsilon)$
is diffeomorphic to the unit cotangent bundle $T^{*}_1S^2$. It follows readily that $m_\gamma \ne 0$, for otherwise
the manifolds $F^{-1}(\gamma_1)$ and $F^{-1}(\gamma_2)$,
where $\gamma_1$ and $\gamma_2$ are the curves shown in Fig.~\ref{fig/ct}, would be diffeomorphic. This is not the case since $F^{-1}(\gamma_1)$ and $F^{-1}(\gamma_2)$ are isotopic to 
$H^{-1}(1 - \varepsilon)$ and $H^{-1}(1 + \varepsilon)$, respectively.
\end{proof}

\begin{figure}[ht]
\begin{center}
\includegraphics[width=0.98\linewidth]{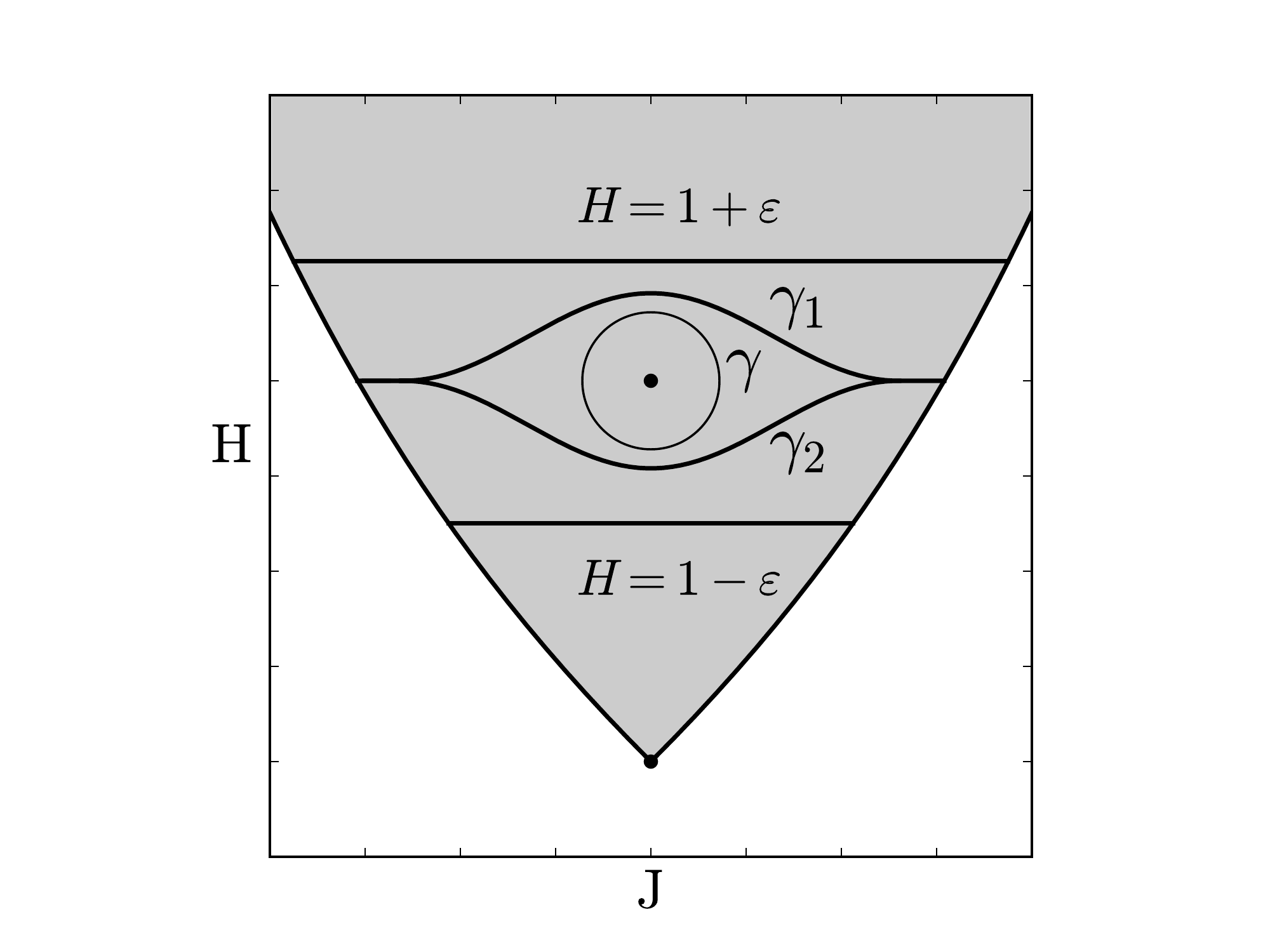}
\caption{Bifurcation diagram for the spherical pendulum, the energy levels, the curves $\gamma_1$ and $\gamma_2$, and the loop $\gamma$ around the focus-focus singularity. The figure is taken from \cite{Martynchuk2020}. }
\label{fig/ct}
\end{center}
\end{figure}

The next step was made by Floris Takens \cite{Takens2010}, who proposed the idea of using Chern numbers of energy hyper-surfaces 
and Morse theory for the computation of monodromy.
More specifically, he observed that  in integrable systems with a Hamiltonian circle action (in particular, in the spherical pendulum), the Chern number of energy hyper-surfaces changes when the energy passes a simple non-degenerate critical value of the Hamiltonian function:

\begin{theorem} \textup{(Takens's index theorem \cite{Takens2010})} \label{theorem/Takens}
Let $H$ be a proper Morse function on an oriented $4$-manifold.
 Assume that $H$ is invariant under a circle action that is free outside the critical points. Let $h_c$ be a critical value of $H$ containing exactly one critical point. Then the Chern numbers 
 of the nearby levels satisfy
 $$c(h_c+\varepsilon) = c(h_c-\varepsilon) 
 \pm 1.
 $$
 Here the sign is plus if the circle action defines the anti-Hopf fibration near the critical point and minus for the Hopf fibration. 
 
\end{theorem}

For the spherical pendulum, the circle action comes from rotational symmetry. The Chern number $c(1+\varepsilon)$
of the energy level $H^{-1}(1 + \varepsilon) \simeq T^{*}_1S^2$ is equal to $2$, and the Chern number
$c(1-\varepsilon)$
of $H^{-1}(1 - \varepsilon) \simeq S^3$ is equal to $1$. Thus, to conclude the proof in this case, it is left to show that 
$m_\gamma = c(1+\varepsilon) - c(1-\varepsilon)$. This last step was made in
\cite{Martynchuk2020}, where it was observed that the monodromy of a two-degree of freedom system with a circle action is given by the difference of the Chern numbers of appropriately chosen energy levels. For the spherical pendulum,
the proof is also based on Fig~\ref{fig/ct}. First, one observes that the Chern number  of 
$F^{-1}(\gamma_1)$ equals to $c(1+\varepsilon) $ and the Chern number  of $F^{-1}(\gamma_1)$ 
to $c(1-\varepsilon).$ The manifolds $F^{-1}(\gamma_1)$ are obtained from solid tori
by gluing the boundary tori via 
$$
\begin{pmatrix}
 a_{-}  \\
 b_{-} 
\end{pmatrix} = \begin{pmatrix}
 1 & c_i \\
 0 & 1
\end{pmatrix} \begin{pmatrix}
 a_{+} \\
 b_{+}
\end{pmatrix},
$$
where $c_i$ is the Chern number of $F^{-1}(\gamma_i)$. (We note that this representation using gluing matrices is a 
very special case of Fomenko-Zieschang theory 
\cite{Fomenko1990, Bolsinov2004}.)
 It follows that the monodromy matrix along $\gamma$ is given by the  product
$$
M_\gamma = \begin{pmatrix}
 1 & c_1 \\
 0 & 1
\end{pmatrix} \begin{pmatrix}
 1 & c_2 \\
 0 & 1
\end{pmatrix}^{-1}.
$$
Since $c_1 - c_2 =  
 1,$ we conclude that the monodromy matrix
$$
M_\gamma = \begin{pmatrix}
 1 & 1 \\
 0 & 1
\end{pmatrix}.
$$

We note that the above Morse-theoretic approach works for more general two-degree of freedom systems that have a global circle action. In particular, one can prove the geometric monodromy theorem using this point of view. 

\subsection{Several remarks}
 There are various cases (systems with many degrees of freedom, non-compact energy levels) when Morse theory cannot be used directly for the computation of monodromy. Nonetheless, as was shown in  \cite{EfstathiouMartynchuk2017, Martynchuk2017}, even in such cases, one can effectively compute the monodromy for integrable systems that are invariant under a global circle action (or a complexity 1 torus action). 
  
The first observation, which is the starting point of the work \cite{EfstathiouMartynchuk2017},  is that in the case of a global circle action, the monodromy of a torus bundle $F \colon F^{-1}(\gamma) \to \gamma$ is given by the Chern number of $F^{-1}(\gamma)$; the Chern number comes from the circle action. More specifically, there is the following result.

\begin{theorem} \textup{(\cite[\S 4.3.2]{Bolsinov2004}, \cite{EfstathiouMartynchuk2017})} \label{theorem/Bolsinov_Fomenko_Zieschang}
Assume that the energy-momentum map $F$ is proper and invariant under a Hamiltonian circle action. 
Let $\gamma \subset \textup{image}(F)$ be  a simple closed curve in the set of the regular values of the map $F$.  Then
 the Hamiltonian
monodromy of the $2$-torus bundle $F \colon F^{-1}(\gamma) \to \gamma$ is given by
$$
\begin{pmatrix}
 1 & m \\
 0 & 1
\end{pmatrix} \in \mathrm{SL}(2,\mathbb Z),
$$
where $m$ is the Chern number of the principal circle bundle $\rho \colon F^{-1}(\gamma) \to  F^{-1}(\gamma) / \mathbb S^1$, which is defined by reducing the circle action.
\end{theorem}

In the case when the curve $\gamma$ bounds a disk $D \subset \textup{image}(F)$,  the Chern number $m$ can be computed from the singularities of the circle action that project into $D$. Specifically,
there is the following result.

 \begin{theorem} \textup{(\cite{EfstathiouMartynchuk2017})}  \label{theorem/EM}
Let $F$ and $\gamma$ be as in Theorem~\ref{theorem/Bolsinov_Fomenko_Zieschang}. Assume that $\gamma$ bounds a $2$-disk $D \subset \textup{image}(F)$ and that
the circle action is free in $F^{-1}(D)$ outside isolated fixed points. Then 
the Hamiltonian monodromy of $F \colon F^{-1}(\gamma) \to \gamma$ is given by 
the number of positive\footnote{The sign of a fixed point depends on whether the circle action defines the anti-Hopf or the Hopf fibration near this point.} fixed points minus the number of negative fixed points in $F^{-1}(D)$. 
\end{theorem}

 We note that Theorems~\ref{theorem/Bolsinov_Fomenko_Zieschang} and \ref{theorem/EM} were generalized to a much more general setting of fractional monodromy and Seifert fibrations; see \cite{Martynchuk2017}. 
 Such a generalization allows one, in particular, to define monodromy for circle bundles over 2-dimensional surfaces of genus $g \ge 1$; in the standard case the genus  $g = 1$. We will come back to fractional monodromy and Seifert manifolds in 
 Section~\ref{sec/fractional_monodromy}.
  
  The works  \cite{EfstathiouMartynchuk2017, Martynchuk2017} essentially settle the monodromy question in the case when the 2 degree of freedom system admits a circle action (or, in the case of many degrees of freedom, a complexity 1 torus action). The case when no such action exists is much less understood. In view of the above Morse theory approach, the following problem seems natural. 
  
  \begin{problem}
 Is it possible to generalise Cushman-Takens approach to the case when there is no Hamiltonian circle action? 
 \end{problem}
 
We note that there are examples of integrable systems with focus-focus fibers and no global circle action; see for example \cite{Waalkens2004,  Leung2003, Tarama2012}. The Hamiltonian monodromy around several such fibers does not have to be of the from 
\begin{equation*} \label{eq/usual_monodromy_matrix}
M_\gamma = \begin{pmatrix}
 1 & k  \\
0 & 1 
\end{pmatrix}.
\end{equation*}
In fact, it can be any $\textup{SL}(2,\mathbb Z)$ matrix (this follows from properties of the group $\textup{SL}(2,\mathbb Z)$); see \cite{Cushman2002, Cushman2002b}. 

In this connection, we mention the class of integrable geodesic flows on \textit{Sol}-manifolds that was constructed in \cite{Bolsinov2006b}. This class comes from a deep problem of non-integrability in classical mechanics  \cite{Bolsinov2000, Kozlov1979, Kozlov1983}. In this case, the monodromy is associated to a degenerate singular fiber, and a $2\times2$ block of the Hamiltonian monodromy matrix is given by an integer hyperbolic matrix. One particular example is 
$$
M_\gamma = \begin{pmatrix}
 2 & 1 & 0 \\
 1 & 1&0 \\
 0 & 0 & 1   
\end{pmatrix}.
$$
We note that cases of such general $\textup{SL}(n,\mathbb Z)$ monodromy matrices (in $n = 2$ or 3 degree of freedom  systems) are not yet understood and new examples are currently missing.
 \begin{problem} \textup{(A. Bolsinov)}
 Construct new examples of integrable systems with a prescribed  monodromy around a (possibly degenerate) singular fiber.  \end{problem}

\subsection{Monodromy in nearly integrable systems}

Let $F \colon M \to \mathbb R^n$ be a proper integral map of an integrable Hamiltonian system on $M$. Assume that the Hamiltonian $H$ is real-analytic and Kolmogorov nondegenerate. Then, according to the Kolmogorov-Arnol'd-Moser theory \cite{Kolmogorov1954, Arnold1963, Moser1967}, there are invariant Liouville tori $F^{-1}(f)$, forming a set of measure $1 - O(\sqrt{\varepsilon})$,  which survive small perturbations $H + \varepsilon P$ of $H$.
This leads to the following natural question, which was addressed in \cite{Rink2004, Broer2007a, Broer2007}, cf. \cite{Zung1996}:  can one extend geometric invariants  of integrable systems (like monodromy) to the nearly-integrable case?
It turns out that this is indeed possible, at least in the topological setting. More specifically, one can `smoothly interpolate' the invariant tori given by the KAM theorem in a global way. Such an interpolation results in a torus bundle for the perturbed system which is diffeomorphic to the original torus bundle associated to $H$. This  implies that the topology of the original torus bundle, given by the non-singular part of $F$, is  preserved under the perturbation. In particular,   Hamiltonian monodromy can be extended to nearly-integrable systems. Below we discuss this idea in more detail, following mainly \cite{Broer2007a}.

Consider the product $D^n \times \mathbb T^n$ of an $n$-disk and an $n$-torus with the standard symplectic structure $dI \wedge d\varphi$. Suppose that $H$ is a non-degenerate Hamiltonian of the integral map $F = \Pr \colon D^n \times \mathbb T^n \to D^n$. This means that the frequency map 
$$\omega_i = \dfrac{\partial H}{\partial I_i} \colon D^n \to \mathbb R^n
$$
is a diffeomorphism onto its image. For $\tau \ge  n$ and $\gamma > 0$, let 
$$D_{\tau, \gamma} = \{ \omega \in \mathbb R^n \mid \langle \omega, k \rangle \ \ge \gamma |k|^{-\tau}, \mbox{ for all } k \in \mathbb Z^n \setminus \{0\}\}$$ be 
the set of Diophantine frequency vectors. We also let
$$A_{\tau, \gamma} = \{I \in D^n \mid \omega(I) \in D_{\tau, \gamma} \ \mbox{ and } \ \textup{dist}(\omega(I), \partial \omega(D^n)) < \gamma\}.$$ 

A main ingredient in the proofs of the monodromy invariance under perturbations is the following (semi-)local theorem of P{\"o}schel \cite{Poschel1982}.

\begin{theorem} \label{theorem/kamlocal} \textup{(Semi-local KAM theorem \cite{Poschel1982}).}
Consider the product $D^n \times \mathbb T^n$ with the standard symplectic structure. Suppose that $H$ is a non-degenerate integral of $F = \Pr \colon D^n \times \mathbb T^n \to D^n$. Let $P$ be a smooth function on $D^n \times \mathbb T^n$. 
Then for all sufficiently small $\varepsilon$, there exists a diffeomorphism $\Phi_{\varepsilon} \colon D^n \times \mathbb T^n \to D^n \times \mathbb T^n$ such that 

(i) $\Phi_{\varepsilon}$ is close to the identity;

(ii) the restriction of $\Phi_{\varepsilon}$ to $A_{\tau, \gamma} \times \mathbb T^n$  conjugates the Hamiltonian flows of $H$ and $H + \varepsilon P$.

\end{theorem}

We note that in integrable systems, the product $D^n \times \mathbb T^n$ appearing in Theorem~\ref{theorem/kamlocal} comes from semi-local action-angle coordinates. This is why this theorem is semi-local. In \cite{Broer2007a}, by using a partition of unity and a convexity argument, this result was extended to the global setting of (possibly non-trivial) Lagrangian torus bundles.  More specifically, there is the following result.

\begin{theorem} \textup{(\cite{Broer2007a})}
Let $F \colon M \to \mathbb R^n$ be the integral map of an integrable  system such that all of the fibers $F^{-1}(f)$ are compact and connected. Suppose that $H$ is a non-degenerate integral of $F$, and let $P$ be a smooth function on $M$. Finally, consider the non-singular part of $F$ over a relatively compact set $R \subset \mathbb R^n$: the $n$-torus bundle 
$$F \colon F^{-1}(R) \to R.$$ 
Then for all sufficiently small $\varepsilon$, there exists a subset $R'_\varepsilon \subset R$ and a diffeomorphism $\Phi_{\varepsilon} \colon F^{-1}(R) \to F^{-1}(R)$ such that

(i) $\Phi_{\varepsilon}$ is close to the identity;

(ii) $R'_\varepsilon$ is nowhere dense in $\mathbb R^n$ and the measure of $R \setminus R'_\varepsilon$ tends to zero when $\varepsilon$ tends to zero;

(iii) the restriction of $\Phi_{\varepsilon}$ to $F^{-1}(R'_\varepsilon)$  conjugates the Hamiltonian flows of $H$ and $H + \varepsilon P$.

\end{theorem}

\begin{remark}
The construction of the global diffeomorphism $\Phi_{\varepsilon}$ is based heavily on the Whitney extension theorem \cite{Whitney1934} and a unicity theorem \cite{Broer2007}, stating that the local  KAM conjugacies provided by 
Theorem~\ref{theorem/kamlocal}
are unique up to a torus translation on the set of Diophantine tori corresponding to the density points of $A_{\tau, \gamma}$.
\end{remark}

\begin{remark}
In the two degree of freedom case of a focus-focus singularity, the important condition of nondegeneracy of $H$ is fulfilled
in a small neighborhood of the focus-focus fiber; \cite{Zung1996}. 
\end{remark}

From this theorem it readily follows that the notion of Hamiltonian monodromy (as well as Duistermaat's Chern class \cite{Duistermaat1980}) can be extended to sufficiently small perturbations $H + \varepsilon P$ of $H$. 

We note that in the two-degree of freedom case of monodromy around a focus-focus singularity, it is essentially sufficient to apply 
only the semi-local theorem of P{\"o}schel by assuming the interpolation diffeomorphism $\Phi_{\varepsilon}$ to be the identity outside a suitably chosen action-angle chart; for details see \cite{Rink2004}.

\section{Quantum monodromy} \label{sec/quantum_monodromy}

Consider an integrable system $F = (f_1, \ldots, f_n)$ on a cotangent bundle $T^*N$, for instance, the spherical pendulum.
Assume for simplicity, that all of the fibers of $F$ are compact and connected.
Since the symplectic form is exact, one can construct semi-local action coordinates via the formula
$$
I_i = \dfrac{1}{2\pi} \int_{\alpha_i} pdq,
$$
where $\alpha_1, \ldots, \alpha_n$ is a family (of bases of) homology cycles on Liouville tori. Different choices of such cycles result in different sets of (semi-local) action coordinates. These sets of semi-local
action coordinates are related 
by a $\textup{SL}(n, \mathbb Z)$ transformation\footnote{In general, different sets of action coordinates are related by a $\textup{SL}(n, \mathbb Z) \ltimes \mathbb R^n$ transformation; note that in our case, the symplectic form is exact.}:
\begin{equation} \label{eq/action_transformation}
(I_1, \ldots, I_n) = M (I'_1, \ldots, I'_n), \ M \in \textup{SL}(n, \mathbb Z).
\end{equation}

Recall that each of the actions is  a function of $F = (F_1, \ldots, F_n)$. Equating 
\begin{equation} \label{eq/action_quantisation}
I_i = \hbar (n_i + \mu_i), \ i = 1, \ldots, n, 
\end{equation}
the actions $I_i$ to integer multiples of the reduced Plank constant (up to the addition of Maslov's correction $\mu_i$), gives a set of points in the $(F_1, \ldots, F_n)$-space. This set of points is called a \textit{semi-classical spectrum} and 
Eq.~\ref{eq/action_quantisation} is the so-called \textit{Bohr-Sommerfeld} or \textit{action quantisation}. We note that the semi-classical spectrum does not depend on the specific choice of the cycles $\alpha_i$ because of 
Eq.~\ref{eq/action_transformation}.
In fact, this set 
locally looks like a regular $\mathbb Z^n$ lattice by the Arnol'd-Liouville theorem. Due to Hamiltonian monodromy,
this does not have to be the case globally; the global lattice may have a defect \cite{Zhilinskii2005, Zhilinskii2010}. Such a defect is always 
present when there is a focus-focus singularity of the system. In particular, it is present in the spherical pendulum
\cite{Cushman1988}.
The presence of the defect can be revealed through the transport of an elementary cell defined by adjacent points of the spectrum; compare with Fig.~\ref{fig/BD_spectrum_pendulum} for the spherical pendulum.

This is the first step towards \textit{quantum monodromy}. One can call the monodromy based on action quantisation \textit{semi-classical},
since it is constructed out of the underlying classical integrable system. 

To get to a purely quantum case, one considers a set of commuting  operators
$\hat F_1, \ldots, \hat F_n$ whose principal symbols define a classical integrable system on $T^{*}M$ as above (see \cite{Vu-Ngoc1999} for more details).
For instance, for the spherical pendulum, 
$$
\hat F_1 = \hat H = -\tfrac12 \hbar^2 \Delta + V
$$
is the corresponding Schr{\"o}dinger operator on $S^2$
and 
$$
\hat F_2 = \hat J = -i\hbar ( x \partial_{y} - y \partial_{x} ).
$$
The main paradigm is that the semi-classical spectrum obtained from the action quantisation gives
an approximation (in terms of $\hbar$) to the  \textit{joint spectrum} 
$$
\sigma(\hat F_1, \ldots, \hat F_n) = \{(\lambda_1, \ldots,  \lambda _n) \in \mathbb R^n \mid \bigcap_{i = 1}^n \ \textup{Ker} (\hat F_i - \lambda_i I) \ne 0\}
$$
of the commuting operators $\hat F_1, \ldots, \hat F_n$. 
In particular,
one can observe a lattice defect also in the purely quantum problem; see Fig.~\ref{fig/BD_spectrum_pendulum}. 

These ideas were originally introduced by Cushman-Duistermaat \cite{Cushman1988}  and Guillemin-Uribe \cite{Guillemin1989} for the spherical pendulum. They were made precise by S. V{\~u} Ng{\d{o}}c in \cite{Vu-Ngoc1999}; see also \cite{ColindeVerdiere1980, Charbonnel1988}.
For more information on the spectral theory of integrable systems, we refer the reader to 
\cite{Pelayo2012}.

\section{Fractional monodromy} \label{sec/fractional_monodromy}

As we have seen in the previous chapters, Hamiltonian monodromy is intimately
related to the singularities of a given integrable  system. However, this invariant is defined 
for the non-singular part 
$$F \colon F^{-1}(R) \to 
R$$
of the possibly singular torus fibration $F \colon M \to \mathbb R^n$ that comes with the  system. 
An invariant that generalises Hamiltonian monodromy to singular torus fibrations was introduced 
by Nekhoroshev, Sadovski\'{i} and Zhilinski\'{i} in 
\cite{Nekhoroshev2006} and it is called \textit{fractional monodromy}.

\subsection{$1$:$(-2)$ resonant system} \label{sec/intrresonance}
Fractional monodromy has up until now been discussed mainly for the so-called $m$:$(-n)$ resonances;
see \cite{Efstathiou2007, Sugny2008, Schmidt2010, Efstathiou2013}.
We shall only focus here on the special case of $1$:$(-2)$ resonance, which is the simplest and historically the first example of an integrable Hamiltonian system with fractional monodromy introduced in the work \cite{Nekhoroshev2006}.

Consider $\mathbb R^4$ with the standard symplectic structure $\omega = dq \wedge dp$. Let the integral map 
$F = (H, J) \colon \mathbb R^4 \to \mathbb R^2 $ be defined by the Hamiltonian function
\begin{equation*}
H = 2q_1p_1q_2 + (q_1^2 - p_1^2)p_2 +  R^2, 
\end{equation*}
where $R = \frac{1}{2}(q_1^2 + p_1^2) + (q_2^2 + p_2^2)$, and the
 `momentum'
\begin{equation*}
J = \frac{1}{2}(q_1^2 + p_1^2) - (q_2^2 + p_2^2).
\end{equation*}
We note that the functions $H$ and $J$ are involution, so that
$F$ is indeed the integral map of an integrable Hamiltonian system. We also note that
the function $J$ defines a  Hamiltonian circle action on $\mathbb R^4$ which preserves the fibration given by $F$.

The bifurcation diagram of the integral map $F$
is shown in Figure~\ref{1stBD}.  From the structure of the diagram we observe that the Hamiltonian monodromy is trivial.
Indeed, the set
$$R  = \{f \in \textup{image}(F) \mid f \mbox{ is a regular value of } F\}$$
is contractible. In particular, every closed path in  $R$ can be deformed to a constant path within $R$. Non-triviality appears if one considers the closed curve $\gamma$ that is shown in Fig.~\ref{1stBD}. 

\begin{figure}[htbp]
  \includegraphics[width=1.1\linewidth]{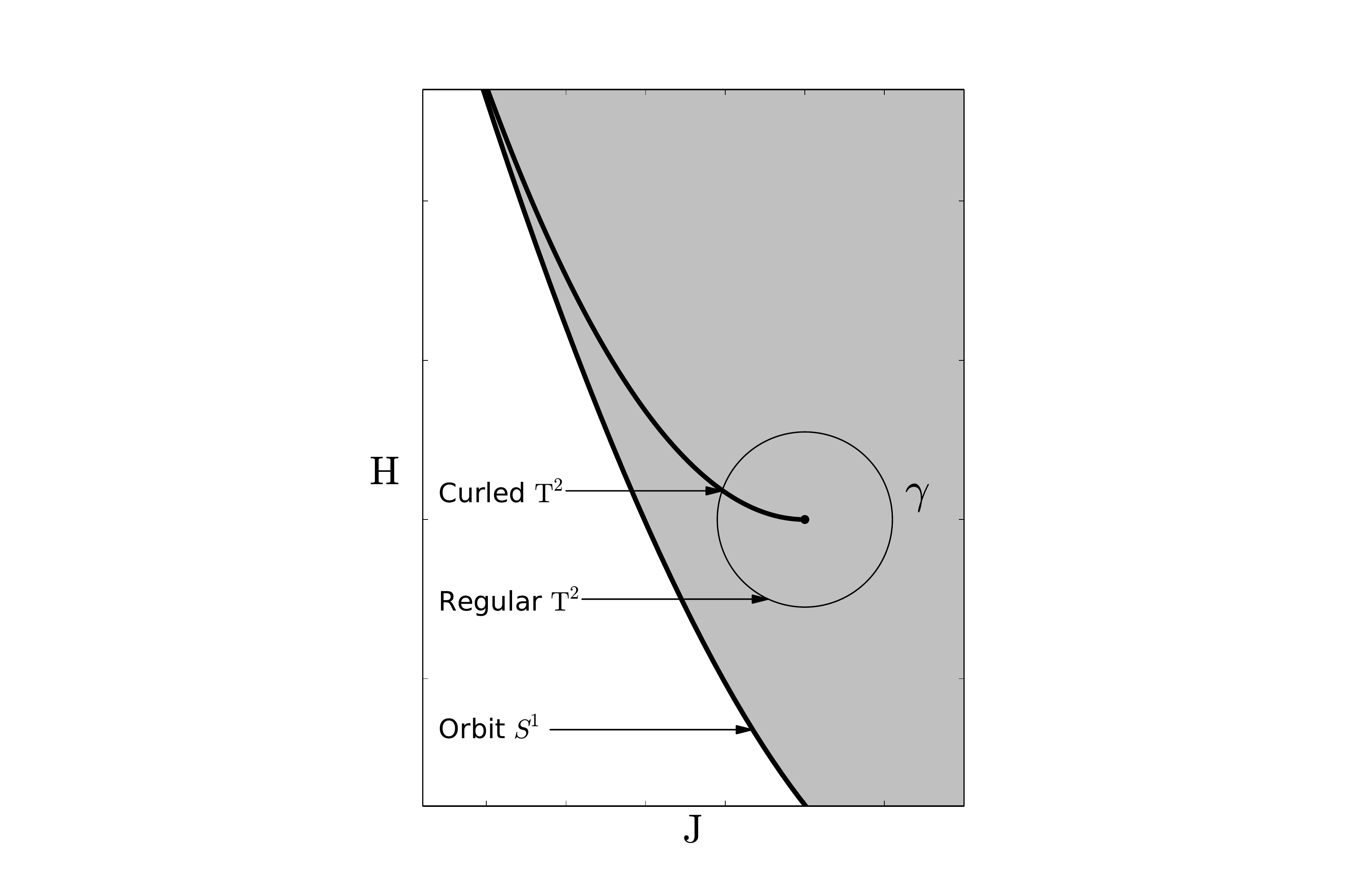}
  \caption{The bifurcation diagram of the integrable $1$:$(-2)$  resonance.
   The closed curve $\gamma$ around the origin intersects the critical hyperbolic branch.}
  \label{1stBD}
\end{figure}

More specifically, consider a non-singular point
$\gamma(t_0)$ and a basis $(a_0,b_0)$ of the integer homology group $H_1(F^{-1}(\gamma(t_0))) \simeq \mathbb Z^2$. Then one can try to `parallel transport' these cycles along $\gamma$ such that at each regular point $\gamma(t)$ they form a basis of $H_1(F^{-1}(\gamma(t_0)))$ and such that the resulting family of cycles is (locally) continuous, also at the critical fiber, corresponding to the intersection of $\gamma$ with the critical hyperbolic branch\footnote{This critical fiber is the so-called curled torus, which can be obtained as follows. Take the direct product of a figure eight and a segment. Identify the upper and the lower boundary components of this product after making a rotation (of the upper component) by the angle $\pi$. The result is schematically shown in Fig.~\ref{Curledtorus}.}. We note that in the case of Hamiltonian monodromy, when we are moving along 
regular Liouville tori, such a parallel transport is always possible \cite{Duistermaat1980}. In this fractional monodromy case, it turns out that 
only a subgroup of $H_1(F^{-1}(\gamma(t_0)))$ can be transported through the critical fiber. Specifically, there is the following result.

\begin{figure}
 \includegraphics[width=0.8\linewidth]{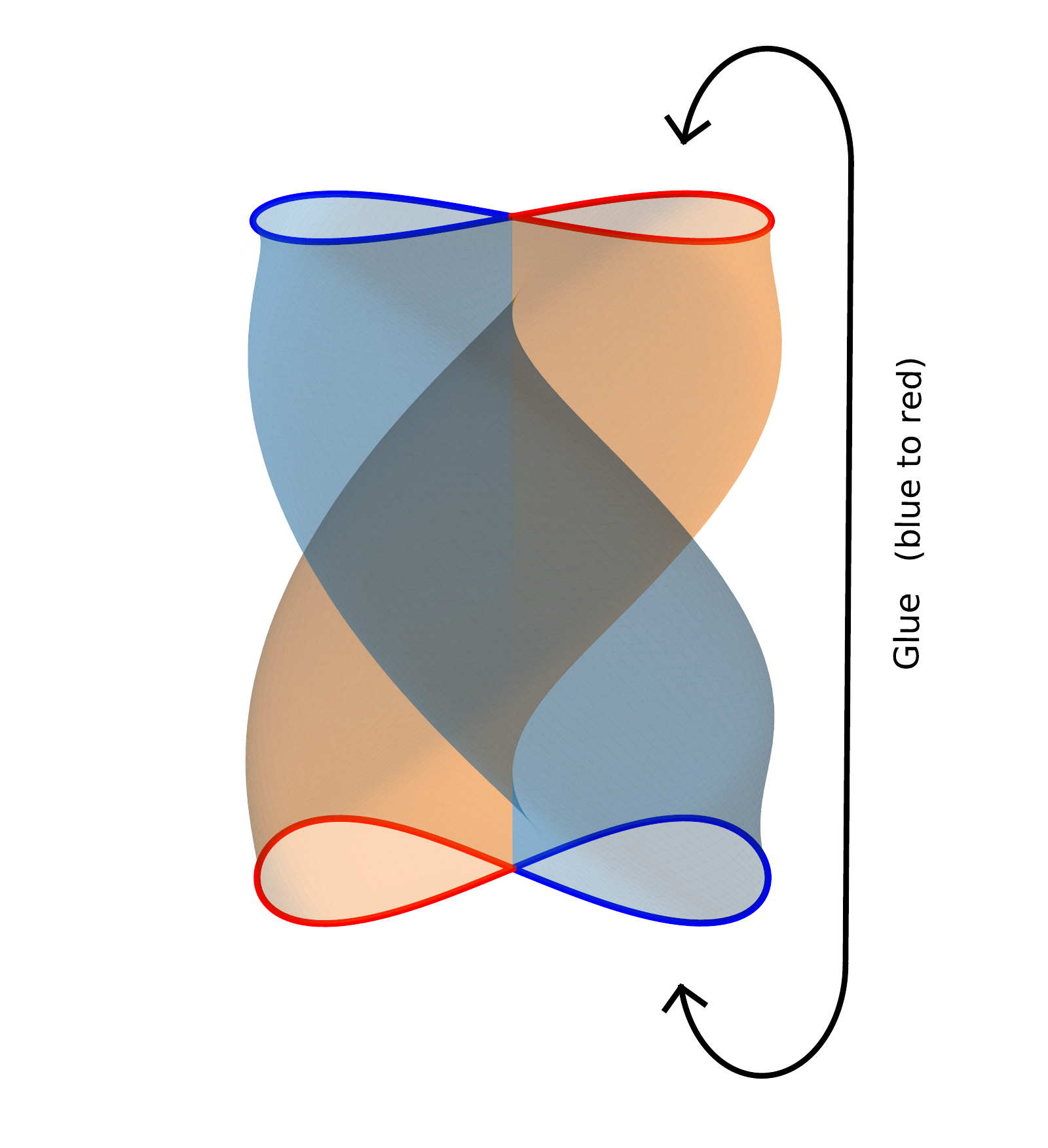}
  \caption{Curled torus}
  \label{Curledtorus}
\end{figure}

\begin{theorem} \textup{(\cite{Nekhoroshev2006})} \label{theorem/fmresonance}
Let $(a_0,b_0)$ be an integer basis of $H_1(F^{-1}(\gamma(t_0)))$, where $\gamma(t_0) \in R$ and  $b_0$ is  an orbit of the circle action. 
The parallel transport (fractional monodromy) along the curve $\gamma$ is given by
 \begin{equation*}2 a_0 \mapsto 2 a_0 + b_0, \  b_0 \mapsto b_0.
 \end{equation*} 
\end{theorem}

\begin{remark} 
When written formally in the integer basis $(a_0,b_0)$, the parallel transport has the form of the \textit{rational} matrix
\[
\begin{pmatrix}
 1 & 1/2 \\
 0 & 1
\end{pmatrix} \in \mathrm{SL}(2,\mathbb Q),
\]
called the matrix of \textit{fractional} monodromy.
\end{remark}

\begin{remark} \label{remark/mark} 
 Theorem~\ref{theorem/fmresonance} is closely related to Fomenko-Zieschang theory. More specifically, to the curve $\gamma$ one can associate its loop molecule, which consists of one atom  $A^*$, corresponding to a neighborhood of the curled torus, and the marks $r = \infty, \ \varepsilon  = 1, \ n = 1$. Fractional monodromy is a function of these invariants, in this case determined by the atom $A^*$ and the $n$-mark; see \cite{Bolsinov2012, Martynchuk2018} for more details. 
 
 \end{remark}

 Since the pioneering work \cite{Nekhoroshev2006}, various proofs of Theorem~\ref{theorem/fmresonance} appeared; see
 \cite{Efstathiou2007, Giacobbe2008, Sugny2008, Broer2010, Tonkonog2013, Efstathiou2013, Martynchuk2017}. A natural approach, which was pursued in \cite{Giacobbe2008, Broer2010, Tonkonog2013},
 is to separate the problem into two parts: the computation of fractional monodromy in a neighborhood $U$ of the curled torus and the computation of (essentially) 
 the usual monodromy outside of this neighborhood $U$. We note that the Liouville fibration inside $U$ is topologically standard (that is, does not depend on the specific system, but only on the singularity). Another  approach, 
 which was pursued in the work \cite{Sugny2008}, is to complexify the system to bypass the hyperbolic branch and compute 
 the variation of the rotation number in the complexified domain; cf. \cite{Audin2002}. We note that this approach works also for higher order resonances. Below we sketch a different proof of Theorem~\ref{theorem/fmresonance}, following the point of view of Seifert manifolds, developed in the work \cite{Martynchuk2017}.

\begin{proof}[Proof of Theorem~\ref{theorem/fmresonance}]
Consider again the curve $\gamma$ shown in Fig.~\ref{1stBD}.
The key observation, which was already made in \cite{Bolsinov2012}, is that $F^{-1}(\gamma)$ is  a Seifert $3$-manifold. The structure of a Seifert fibration comes from the circle action given by the
momentum $J$.
In complex coordinates
$z = p_1 + iq_1$ and $w = p_2+iq_2$, this circle action has the form
\begin{equation} \label{circleaction}
(t,z,w) \mapsto (e^{it}z, e^{-2it}w), \ t \in \mathbb S^1.
\end{equation}
We observe that the origin is fixed under this action and that the set
$$
P = \{(q,p) \mid q_1 = p_1 = 0 \mbox{ and } q_2^2 + p_2^2 \ne 0\}
$$
consists of points with $\mathbb Z_2$ isotropy group. This implies that the Euler number of the Seifert
manifold $F^{-1}(\gamma)$ equals $1/2 \ne 0$. Indeed, Stokes' theorem implies that
the Euler number of
$F^{-1}(\gamma)$ coincides with the Euler number of a small $3$-sphere around the origin $z = w = 0$. The latter Euler number equals $1/2$
because of  \eqref{circleaction}. From this and Theorem~\ref{theorem/Bolsinov_Fomenko_Zieschang}, we get the following.

\begin{lemma} \label{lemma/fractm_quotient} \textup{(\cite{Martynchuk2017})}
 The quotient space
 $F^{-1}(\gamma) / \mathbb Z_2$ is the total space of a torus bundle over $\gamma$. Its monodromy is given by
 $$
M = \begin{pmatrix}
 1 & 1 \\
 0 & 1
\end{pmatrix} \in \mathrm{SL}(2,\mathbb Z).
$$
\end{lemma}

From Lemma~\ref{lemma/fractm_quotient} we infer that the 
parallel transport along the curve $\gamma$ in the $\mathbb Z_2$-quotient space has the form 
 \begin{equation*}a^r_0 \mapsto a^r_0 + b^r_0, \  b^r_0 \mapsto b^r_0,
 \end{equation*} 
 where the cycles $a^r_0 = a_0 /\mathbb Z_2$ and $b^r_0 = b_0 /\mathbb Z_2$  form the induced basis of the group
 $H_1(F^{-1}(\gamma(t_0)) / \mathbb Z_2).$ 
 Observe that $a_0$ is not affected by the quotient map, and the orbit $b_0$ becomes `shorter': $2 b^r_0 \simeq b_0$.
 It follows that the parallel transport in the original space 
 has the form
 \begin{equation*}2a_0 \mapsto 2a_0 + b_0, \  b_0 \mapsto b_0.
 \end{equation*} 
This concludes the proof of Theorem~\ref{theorem/fmresonance}.
\end{proof}

We note that the idea of computing fractional monodromy using a covering map appeared in the work \cite{Efstathiou2013}, where the authors computed fractional monodromy for a large class of integrable systems with an $m$:$(-n)$ resonance.
There an \textit{un}covering map was used to lift the (possibly singular) Lagrangian fibers to a union of tori. Here we used a covering map instead. Moreover, we focused not on the fibers of the energy-momentum map, but rather on
the global topology of an associated Seifert fibration. This approach, which was  developed in the work \cite{Martynchuk2017}, turned out to be very effective and allowed one to define fractional monodromy over an arbitrary Seifert manifold  with an orientable 
base of genus $g \ge 1$. (We note that the known examples appeared as a special case of this construction when 
the genus $g = 1$ and there are at most two singular fibers of the Seifert fibration.) The precise results can be stated as follows; cf. Theorems~\ref{theorem/Bolsinov_Fomenko_Zieschang} and \ref{theorem/EM}.

\begin{theorem} \label{theorem/fracm_main} \textup{(\cite{Martynchuk2017})}
  Let $X$ be the total space of a Seifert fibration  with an orientable 
base such that the boundary of $X$ consists of two tori. Let $X_f$ be the closed Seifert manifold obtained by gluing these tori via a fiber-preserving diffeomorphism $f$. Take bases of these tori $(a_0,b_0)$ and $(a_1,b_1)$ such that $b_0, b_1$ correspond to non-singular fibers of the Seifert fibration.  Let $N$ denote the least common multiple of the orders of exceptional fibers. Then only linear combinations of $Na_0$ and $b_0$ can be parallel transported along $X$ and under the parallel
  transport
  \begin{equation*}
 \begin{split}
N a_0 & \mapsto N a_1 + k b_1 \\
b_0 & \mapsto b_1
\end{split}
\end{equation*} 
for some integer $k = k(f)$ which depends only on the isotopy class of the diffeomorphism $f.$ Moreover, the Euler number of $X(f)$ is given by 
$e(f) = k(f)/N.$
\end{theorem}

\begin{remark}
We note that, in this case, the matrix of fractional monodromy is given by 
$$M_X = 
\begin{pmatrix}
 1 & e(f) \\
 0 & 1
\end{pmatrix} \in SL(2,\mathbb Q).
$$
\end{remark}

\begin{remark}
In Theorem~\ref{theorem/fracm_main}, we use the notion of parallel transport introduced in \cite{Efstathiou2013}. Specifically, let $\partial X = \mathbb T^2_1 \sqcup \mathbb T^2_0$. By definition, a cycle $\alpha_1 \in H_1(\mathbb T^2_1)$ is a parallel transport of $\alpha_0 \in H_1(\mathbb T^2_0)$ if these cycles are of the same integer homology class in $X$.
We note that this definition of parallel transport can be used for abstract manifolds  with boundary, without an explicit connection to integrability. However, 
such a parallel transport is not always well defined: one can construct examples of $3$-manifolds where parallel transport is not unique or does not give rise to a well-defined automorphism \cite{Martynchuk2018}.
According to Theorem~\ref{theorem/fracm_main}, this notion of parallel transport is well defined for Seifert manifolds with an orientable base and results in an automorphism of an index-$N$ subgroup of $H_1(\mathbb T^2_0 \simeq_f \mathbb T^2_1).$
\end{remark}

Theorem~\ref{theorem/fracm_main} implies that in order to compute fractional monodromy for a specific integrable system, it is  sufficient to compute the orders of exceptional orbits  and the Euler number of the corresponding Seifert fibration. We note that in concrete examples of integrable systems, the orders of exceptional orbits  are often known from the circle action.  To compute the Euler number,  one can use the following result.

\begin{theorem}  \textup{(\cite{Martynchuk2017})}
Let $M$ be a compact oriented $4$-manifold that admits an effective circle action. Assume that the action is fixed-point free on the boundary
$\partial M$ and has only finitely many fixed points $p_1, \ldots, p_{\ell}$ in the interior. Then
\begin{equation*}
e(\partial M) = \sum\limits_{k=1}^{\ell} \dfrac{1}{m_k n_k},
\end{equation*} 
where $(m_k,n_k)$ are isotropy weights of the fixed points $p_k$.
\end{theorem}

We note that the idea of using Seifert fibration in the context of integrable systems goes back to A.~T. Fomenko and H. Zieschang. In their molecule theory \cite{Fomenko1990, Bolsinov2004}, atoms and Seifert manifolds appear as the basic building blocks. However, not every loop molecule admits the structure of a global Seifert fibration. 

\begin{problem}(\textup{A.T. Fomenko})
Suppose that $X$ corresponds to a loop molecule of an integrable and non-degenerate two-degree of freedom system. Then $X$ admits a decomposition into Seifert-fibered pieces. Can one construct an algorithm that computes fractional monodromy of $X$, when it exists?
\end{problem}

A related problem is the following.

\begin{problem}
Suppose $X$ is a graph-manifold (a loop molecule). Under which geometric conditions does fractional monodromy exist along $X$?
\end{problem}

\subsection{Towards quantum fractional monodromy}

Let us come back to the example of a system with $1$:$(-2)$ resonance. Consider the (semi-local) action coordinates
\begin{align*}
I_1 &= \dfrac{1}{2\pi} \int_{\alpha_1} pdq  \ \mbox{ and} \\
I_2 &= \dfrac{1}{2\pi} \int_{\alpha_2} pdq,
\end{align*}
where the cycle $\alpha_2$ corresponds to the circle action and $\alpha_1$ is such that
$(\alpha_1, \alpha_2)$ form a basis in the first homology group of a Liouville torus. Note that $I_2 = J$.

As in the case of the usual quantum monodromy, one can
consider
 the quantisation condition
\begin{align*}
I_1 &= \hbar (n_1 + \mu_1), \\
I_2 &= \hbar (n_2 + \mu_2),
\end{align*}
which gives a semi-classical spectrum locally outside the hyperbolic branch. However,
for this spectrum one cannot transport an elementary cell around the singularity
in a continuous way. The novel idea that was introduced in \cite{Nekhoroshev2006}
is to consider not an elementary cell, but a double cell in this case. Let us explain this idea on the level of the actions. Observe that by
Theorem~\ref{theorem/fmresonance}, it is possible to define 
$2I_1$ and $I_2$ also in a neighborhood of the curled torus. Therefore,
the action quantisation
\begin{align*}
2 I_1 &= \hbar (n_1+2\mu_1), \\
I_2 &= \hbar (n_2+\mu_2)
\end{align*} 
will result in a globally defined lattice which is contained in the original semi-classical spectrum and for which one can transport 
an elementary cell around the origin. By the construction, an elementary cell for this lattice is a double cell for the original spectrum.

We note that here we suppress the question of a continuous transport of an elementary cell in the joint spectrum of $(\hat H, \hat J)$ for the quantum $1$:$(-2)$ resonance system near the hyperbolic branch. 

The idea of considering a double or an $n$-cell leads to the notion of \textit{quantum} \textit{fractional monodromy} \cite{Nekhoroshev2006}. 
We refer the reader to \cite{Nekhoroshev2006} for more details.

\section{Scattering monodromy}
  
Up until now we considered integrable Hamiltonian systems such that the corresponding integral map $F$ has compact invariant fibers $F^{-1}(f), \ f \in \mathbb R^n$. In this section, we mainly discuss the non-compact case. In particular, we discuss the so-called \textit{scattering monodromy} in the context of classical
potential scattering theory. 

\subsection{Preliminaries}

A notion of scattering monodromy was originally introduced by L.~M. Bates and R.~H. Cushman in \cite{Bates2007} for a two degree of freedom hyperbolic oscillator \footnote{The hyperbolic oscillator is not a scattering system in the sense of, for instance, \cite{Knauf1999}, since the potential
of this system is unbounded at infinity and is not decaying to zero. Nonetheless, the system shares some of the properties of scattering systems, such as the existence of the so-called \textit{deflection angle}; see below.}.
At about the same time, 
scattering monodromy was introduced by H.~R. Dullin and H. Waalkens in \cite{DullinWaalkens2008} for planar scattering systems with 
a repulsive rotationally symmetric potential, both in the classical and quantum settings. The  idea behind the works 
\cite{Bates2007, DullinWaalkens2008} is as follows.

Consider a Hamiltonian system on $T^{*}\mathbb R^2$ with canonical coordinates 
$(q_1,q_2,p_1,p_2)$ defined by the Hamiltonian function
$$
H = \dfrac{1}{2}(p_1^2+p^2_2) - \dfrac{1}{2}V(r),
$$
where $V$ is a radially symmetric potential $r^2 = q_1^2+q_2^2$.  This system describes the motion of a particle on the plane $\mathbb R^2$ with coordinates $(q_1,q_2)$ under the influence 
of the potential function $V$. We observe that the system is Liouville integrable since the momentum $J = q_1p_2-q_2p_1$ is conserved. 

We shall assume, for simplicity, that the potential $V$ is a smooth, monotone function, decaying at infinity sufficiently fast. 
The bifurcation diagram of $F = (H,J)$ is shown in Fig.~\ref{fig/BD_scattering}. It consists of a single critical value, corresponding 
the maximum of $V$. This is a focus-focus singularity if the maximum is non-degenerate. In particular,
the set $R$ of the regular values of $F$ is not simply-connected. Nonetheless, it can be shown that global action-angle coordinates exist for this system; see \cite{Bates2007}. Topologically, the bundle $F^{-1}(\gamma) \to \gamma$ is a trivial cylinder $S^1\times \mathbb R$-bundle. Moreover, the energy levels $H^{-1}(h_{max} \pm \varepsilon)$
below and above $h_{max} = \max V$ are topologically the same.

\begin{figure}[ht]
\begin{center}
\includegraphics[width=0.9\linewidth]{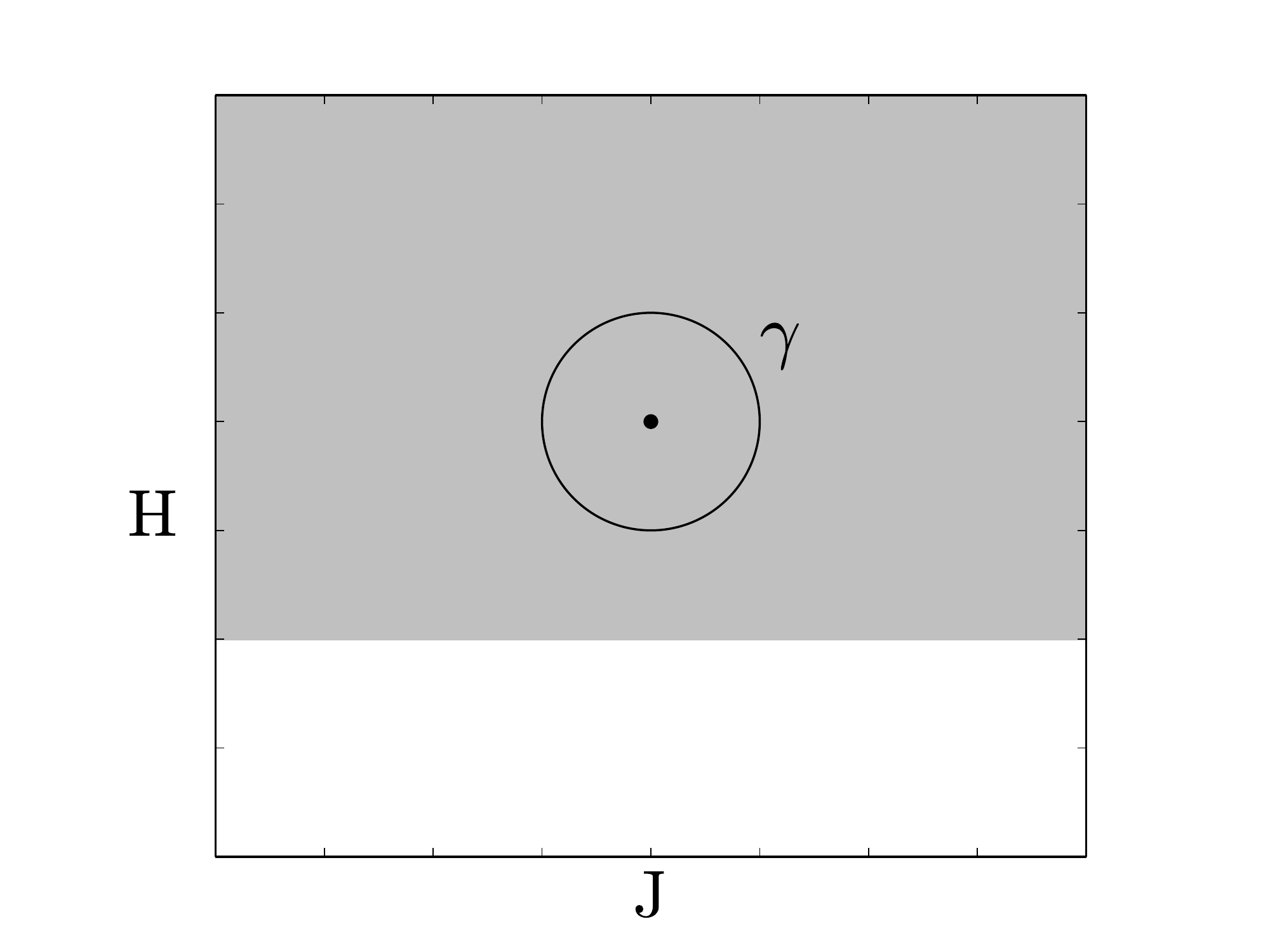}
\caption{Bifurcation diagram for the integral map $F = (H,J)$ of a planar scattering integrable system with a repulsive potential. }
\label{fig/BD_scattering}
\end{center}
\end{figure}

To get a non-trivial invariant, the authors of \cite{Bates2007, DullinWaalkens2008} considered the so-called deflection angle of a trajectory. Specifically,
observe that under the Hamiltonian dynamics, a particle in the plane gets deflected by $V$. It proceeds to spatial infinity in 
both forward and backward time, unless it approaches the maximum of the potential. To any such scattering trajectory, one can associate the deflection angle
$$\Phi = \dfrac{1}{2\pi}\int\limits_{-\infty}^{+\infty} \dfrac{d\varphi(q(t))}{dt} dt,$$
where $\varphi$ is the polar angle in the configuration $q_1q_2$-plane. Due to rotational symmetry, the deflection
angle is a function of $F = (H,J)$. Hence, one can consider its variation along $\gamma$.
\begin{theorem} \textup{(\cite{Bates2007, DullinWaalkens2008})}
In the above setting, the variation of the deflection angle $\Phi$ along $\gamma$ is equal to $-1$.
\end{theorem}

The above approach to scattering monodromy is based on the notion of a deflection angle, which is very close to the notion of a rotation number for compact systems. We note that one can approach scattering monodromy also from other (related) perspectives. For instance,
in \cite{DullinWaalkens2008} the authors used radial actions for the pair of integrable systems: the original system given by $V$ and a reference system with the zero potential (the free flow). These radial actions 
$$
I = \dfrac{1}{\pi} \int^{\infty}_{r_0} p_rdr \ \mbox{ and } \ I_{\mathrm{ref}} = \dfrac{1}{\pi} \int^{\infty}_{r'_0} p^{\mathrm{ref}}_rdr
$$
do not exist individually. However, if the potential $V$ decays sufficiently fast, their difference exists. More specifically,
the limit
\begin{equation*}
\lim_{r \to \infty} \dfrac{1}{\pi} \int^R_{r_0} p_rdr - \dfrac{1}{\pi} \int^R_{r'_0} p^{\mathrm{ref}}_rdr
\end{equation*}
exists and behaves like a usual radial action of a compact system with a rotationally-symmetric potential. In particular, transporting this radial action and the action
$J$ along $\gamma$, one gets a monodromy automorphism of the usual form:
\begin{equation*}
 M_\gamma = \begin{pmatrix} 1 & m_{\gamma} \\ 0 & 1\end{pmatrix},
 \end{equation*}
where $m_{\gamma} = -1$ is the variation of the deflection angle. 

Related to this is a `billiard' approach, which is also based on the action coordinates. It is applicable whenever a given integrable system with non-compact fibers is separable. We refer the reader to the works \cite{Delos2008, Sadovski2016, Meesters2017}.

We also mention the work \cite{Efstathiou2017}, where the notion of \textit{non-compact monodromy} was introduced. Here the idea is that for a non-compact integrable system with the integral map $F$ and a global circle action, one can compactify the fibers of $F$ near a focus-focus fiber preserving the circle action. Then one gets a compact fibration with the usual monodromy around the focus-focus fiber. In \cite{Efstathiou2017},  this monodromy is called 
non-compact. It coincides with the scattering monodromy for the above two-degree of freedom systems.

Finally, we mention the work \cite{Martynchuk2019}, where the authors follow the point of view of classical potential scattering theory; see, in particular, \cite{Knauf1999}. The novelty of this work is that it is applicable to possibly many degrees of freedom scattering and integrable systems that are not necessarily rotationally symmetric. This approach generalises the above approaches to scattering monodromy. We discuss it in some more detail below.

\subsection{Classical scattering theory}

Below we briefly review classical potential scattering theory, following mainly A. Knauf \cite{Knauf1999, Knauf2011} and J. Derezinski and C. Gerard \cite{Derezinski2013}; see also \cite{Martynchuk2018, Martynchuk2019}.

Consider a pair of  Hamiltonians  on $T^{*}\mathbb R^n$ given by
$$H = \frac{1}{2} \|p\|^2  + V(q) \mbox{ \ and \ } H_r = \frac{1}{2} \|p\|^2+V_r(q),$$
where the (singular) potentials $V$ and $V_r$ are  assumed to decay sufficiently fast.
Let $g_H^t$ denote the Hamiltonian flow. Define the invariant set $s$ of scattering states by
$$s = \{(q,p) \in T^{*}\mathbb R^n \mid H(q,p) > 0, \ \textup{sup}_{t \in \mathbb R^{\pm}} \|g_H^t(q,p)\| = \infty\}.$$

If the potential  $V$ decays at infinity sufficiently fast (for example, is of short range \cite{Knauf1999, Derezinski2013}), then the trajectories are asymptotic to straight lines. Moreover, for any $x \in s$,
the following functions, usually called the \textit{asymptotic direction}
and the \textit{impact parameter} of $g_H^t(x)$, 
\begin{equation*}
 {\hat{p}}^{\pm}(x) = \lim_{t\to \pm \infty}  p(t,x) \  \mbox{ and } \ 
 q_{\bot}^{\pm}(x) = \lim_{t\to \pm \infty}  q(t,x) - \langle q(t,x),  {\hat{p}}^{\pm}(x)\rangle \frac{{\hat{p}}^{\pm}(x)}{2h},
\end{equation*}
are defined and  depend continuously on $x \in s$. (Here $h$ is the energy of $g^t_H$.) In other words, the space of trajectories $s / g^t_H$, that is, 
the quotient space of $s$ with respect to the Hamiltonian flow $g^t_H$, gets parametrised by the trajectories of the free Hamiltonian $H = \frac{1}{2}\|p\|^2.$
 Due to the  $g_H^t$-invariance, we get the maps
$$
A^{\pm} = (\hat{p}^{\pm},q_{\bot}^{\pm}) \colon  s / g^t_H \to  AS 
$$
from $s / g^t_H$ to a subset $AS \subset \mathbb R^{n}\times\mathbb R^n$ of the `asymptotic states'. 

Similarly, one can construct
the maps 
$$
A_r^{\pm} = (\hat{p}^{\pm},q_{\bot}^{\pm}) \colon  s_r / g^t_{H_r} \to AS
$$
for the Hamiltonian $H_r =  \frac{1}{2}p^2 + V_r(q).$ 

\begin{definition} \textup{(\cite{Knauf1999,Martynchuk2019})} Let $M$ be a $g_H^t$-invariant submanifold of $s$. Assume that the composition map
 $$
 S = (A^{-})^{-1} \circ A^{-}_r \circ (A^{+}_r)^{-1} \circ A^{+}
 $$
is well defined and maps the set $B = M / g_H^t$ to itself. The map $S$ is then called the \textit{scattering map} with respect to $H, H_r$ and $B$.
\end{definition}

\subsection{Monodromy in scattering  systems}

To define scattering monodromy, we need to restrict the class of possible reference systems to those for which the corresponding scattering map preserves 
the integral fibration at infinity.

\begin{definition}   \textup{(\cite{Martynchuk2019})}
Consider a Hamiltonian $H$ which gives rise to a scattering integrable  system with the integral map $F$. A  Hamiltonian $H_{r}$ will be called a 
\textit{reference} Hamiltonian for this system if 
\begin{equation} \label{eq/asymptotic_condition}
 F\left(\lim\limits_{t\to+\infty} g^t_{H_{r}}(x)\right) = F\left(\lim\limits_{t\to-\infty}g^t_{H_{r}}(x)\right)
\end{equation}
 for every scattering trajectory $t \mapsto g^t_{H_{r}}(x)$.
\end{definition}

\begin{remark} 
We note that
Eq.~\eqref{eq/asymptotic_condition} appeared in a related context in the work \cite{Jung1993}.  
 \end{remark}

Consider the Liouville fibration $F \colon s \to \mathbb R^n$.  Let $H_r$ be a reference Hamiltonian for $F$ such that $A^{\pm}(s) \subset A^{\pm}(s_r)$ holds. Then we have the scattering map 
$$
S \colon B \to B, \ B = s / g^t_H,
$$
which allows us to identify the asymptotic states of $s$ at $t = + \infty$ and $t = -\infty$. This results in a new total space 
$s_c$ and a new fibration
$$
F_c \colon s_c \to \mathbb R^n.
$$

\begin{definition} \label{definition/scattering_monodromy} \textup{(\cite{Martynchuk2019})}
Assume that the fibration
$$F_c \colon s_c \to \mathbb R^n$$ is a torus bundle. The Hamiltonian monodromy of this  bundle is called \textit{scattering monodromy} of $F$ with respect to $H_r$.
\end{definition}

One distinctive  property of scattering monodromy in the sense of Definition~\ref{definition/scattering_monodromy}
is its relative form (dependence on the choice of $H_r$).
For instance, if we choose $H_r$ to coincide with the original Hamiltonian $H$, Duistermaat's Hamiltonian monodromy is recovered. 

Another  property that we mention here is that using an appropriately chosen scattering map, one can define scattering monodromy for certain scattering systems that are not necessarily integrable or even nearly integrable. This is similar to the case of another scattering invariant (the so-called \textit{scattering degree}) introduced by A. Knauf in \cite{Knauf1999}
outside the context of integrability; cf. also the work \cite{Martynchuk2016}.

\subsection{Example}

Let us come back to the example considered at the beginning of this section: a Hamiltonian system on $T^{*}\mathbb R^2$ given by the Hamiltonian function
$$
H = \dfrac{1}{2}(p_1^2+p^2_2) - \dfrac{1}{2}V(r),
$$
where $V$ is a radially symmetric, monotone decaying potential.  Let  $J = q_1p_2-q_2p_1$ denote the angular momentum. Consider the curve $\gamma$ around the focus-focus fiber shown in Fig.~\ref{fig/BD_scattering}. Setting $H_r = \dfrac{1}{2}(p_1^2+p^2_2)$ and $M = F^{-1}(\gamma)$, we get the scattering map
$$
S \colon B \to B, \ B = M/g^t_H.
$$
Note that the manifold $B$ is a two-torus in this case.
\begin{theorem}  \textup{(\cite{Martynchuk2019})}
 In the first homology group of $B= F^{-1}/g^t_H,$ the scattering map $S$ is given by the matrix
  $$
 M_\gamma = \begin{pmatrix} 1 & 1 \\ 0 & 1\end{pmatrix}.
 $$
 This scattering monodromy along $\gamma$ (w.r.t. $H$ and $H_r =  \frac{1}{2}(p_1^2+p^2_2)$) is given by the same
 matrix $M_\gamma.$
\end{theorem}

Another interesting example, where a natural choice of $H_r$ is not given by the free flow, is the (spatial) Euler two-centre problem. We refer to 
the work \cite{Martynchuk2019} for details. 

\subsection{Quantum scattering monodromy}

We have already noted that for a scattering system on $T^*\mathbb R^2$ with a decaying rotationally symmetric potential $V(r)$, one can define a notion of scattering monodromy using the difference of the radial actions
\begin{equation}  \label{eq/action_difference}
I_{\mathrm{diff}} = \lim_{r \to \infty} \dfrac{1}{\pi} \int^R_{r_0} p_rdr - \dfrac{1}{\pi} \int^R_{r'_0} p^{\mathrm{ref}}_rdr,
\end{equation}
for the original system and the reference system with zero potential (the free flow); see \cite{DullinWaalkens2008}. 
Using this idea, it was shown in the same work \cite{DullinWaalkens2008} that for scattering systems in the plane, one can define a quantum analogue of scattering monodromy.  The non-triviality of this invariant also leads to a lattice defect, similarly to the compact case.

We note, however, that in quantum scattering (and even in the case of scattering in the plane), there is an additional difficulty related to the decay of the potential function: if the potential
$V$ is of long range, then the corresponding action difference given in Eq.~\eqref{eq/action_difference} diverges. This is not a problem for the classical scattering monodromy (in the sense of Definition~\ref{definition/scattering_monodromy}).
 Another interesting and related problem is to define quantum scattering monodromy for scattering integrable systems with many degrees of freedom.
For a discussion of these problems, we refer the reader to \cite{Martynchuk2018}.

\bibliographystyle{amsplain}
\bibliography{./library.bib}

\providecommand{\bysame}{\leavevmode\hbox to3em{\hrulefill}\thinspace}
\providecommand{\MR}{\relax\ifhmode\unskip\space\fi MR }
\providecommand{\MRhref}[2]{%
  \href{http://www.ams.org/mathscinet-getitem?mr=#1}{#2}
}
\providecommand{\href}[2]{#2}
\begin{thebibliography}{10}

\bibitem{Arnold1963}
V.~I. Arnol'd, \emph{Proof of a theorem of {A}.~{N}.~{K}olmogorov on the
  invariance of quasi-periodic motions under small perturbations of the
  {H}amiltonian}, Russian Mathematical Surveys \textbf{18} (1963), no.~5,
  9--36.

\bibitem{Arnold1978}
\bysame, \emph{Mathematical methods of classical mechanics}, Graduate Texts in
  Mathematics, vol.~60, Springer-Verlag, New York-Heidelberg, 1978, Translated
  by K. Vogtmann and A. Weinstein.

\bibitem{Arnold1968}
V.~I. Arnol'd and A.~Avez, \emph{Ergodic problems of classical mechanics}, W.A.
  Benjamin, Inc., 1968.

\bibitem{Arnold2012}
V.~I. Arnold, S.~M. Gusein-Zade, and A.~N. Varchenko, \emph{Singularities of
  differentiable maps, volume 2: Monodromy and asymptotics of integrals},
  Modern Birkh{\"a}user Classics, Birkh{\"a}user Boston, 2012.

\bibitem{Audin2002}
M.~Audin, \emph{{H}amiltonian monodromy via {P}icard-{L}efschetz theory},
  Communications in Mathematical Physics \textbf{229} (2002), no.~3, 459--489.

\bibitem{Bates2007}
L.~Bates and R.~Cushman, \emph{Scattering monodromy and the {A}1 singularity},
  Central European Journal of Mathematics \textbf{5} (2007), no.~3, 429--451.

\bibitem{Bates1991}
L.~M. Bates, \emph{Monodromy in the champagne bottle}, Journal of Applied
  Mathematics and Physics (ZAMP) \textbf{42} (1991), no.~6, 837--847.

\bibitem{Bates1993}
L.~M. Bates and M.~Zou, \emph{Degeneration of {H}amiltonian monodromy cycles},
  Nonlinearity \textbf{6} (1993), no.~2, 313--335.

\bibitem{Beukers2003}
F.~Beukers and R.~H. Cushman, \emph{The complex geometry of the spherical
  pendulum}, Celestial mechanics: dedicated to {D}onald {S}aari for his 60th
  birthday (A.~Chenciner, R.~H. Cushman, C.~Robinson, and Z.~Xia, eds.),
  Contemporary Mathematics, vol. 292, American Mathematical Society, 2002,
  pp.~47--70.

\bibitem{Bolsinov2019}
A.~Bolsinov and A.~Izosimov, \emph{Smooth invariants of focus-focus
  singularities and obstructions to product decomposition}, Journal of
  Symplectic Geometry \textbf{17} (2019), no.~6, 1613--1648.

\bibitem{Bolsinov2012}
A.V. Bolsinov, Izosimov A.M., A.Y. Konyaev, and A.A. Oshemkov, \emph{Algebra
  and topology of integrable systems. {R}esearch problems (in {R}ussian)},
  Trudy Sem. Vektor. Tenzor. Anal. \textbf{28} (2012), 119--191.

\bibitem{Bolsinov2006b}
A.V. Bolsinov, H.~Dullin, and A.~Veselov, \emph{Spectra of sol-manifolds:
  Arithmetic and quantum monodromy}, Communications in Mathematical Physics
  \textbf{264} (2006), 588--611.

\bibitem{Bolsinov2004}
A.V. Bolsinov and A.T. Fomenko, \emph{Integrable {H}amiltonian {S}ystems:
  {G}eometry, {T}opology, {C}lassification}, CRC Press, 2004.

\bibitem{Bolsinov2015}
A.V. Bolsinov, A.A. Kilin, and A.O. Kazakov, \emph{Topological monodromy as an
  obstruction to {H}amiltonization of nonholonomic systems: Pro or contra?},
  Journal of Geometry and Physics \textbf{87} (2015), 61--75, Finite
  dimensional integrable systems: on the crossroad of algebra, geometry and
  physics.

\bibitem{Bolsinov2006}
A.V. Bolsinov and A.A. Oshemkov, \emph{Singularities of integrable
  {H}amiltonian systems}, In: Topological Methods in the Theory of Integrable
  Systems, Cambridge Scientific Publ., 2006.

\bibitem{Bolsinov2000}
V.V. Bolsinov and I.A. Taimanov, \emph{Integrable geodesic flows with positive
  topological entropy}, Invent. Math. \textbf{140} (2000), 639--650.

\bibitem{Broer2007a}
H.~W. Broer, R.~H. Cushman, F.~Fass\`{o}, and F.~Takens, \emph{Geometry of
  {K}{A}{M} tori for nearly integrable {H}amiltonian systems}, Ergodic Theory
  and Dynamical Systems \textbf{27} (2007), no.~3, 725--741.

\bibitem{Broer2007}
H.~W. Broer and F.~Takens, \emph{Unicity of {K}{A}{M} tori}, Ergodic Theory and
  Dynamical Systems \textbf{27} (2007), no.~3, 713--724.

\bibitem{Broer2010}
H.W. Broer, K.~Efstathiou, and O.V. Lukina, \emph{A geometric fractional
  monodromy theorem}, Discrete and Continuous Dynamical Systems \textbf{3}
  (2010), no.~4, 517--532.

\bibitem{Charbonnel1988}
A.M. Charbonnel, \emph{Comportement semi-classique du spectre conjoint
  d?op{\'e}rateurs pseudo-diff{\'e}rentiels qui commutent}, Asymptotic Analysis
  \textbf{1} (1988), 227--261.

\bibitem{ColindeVerdiere1980}
Y.~Colin~de Verdi{\'e}re, \emph{Spectre conjoint d?op{\'e}rateurs
  pseudo-diff{\'e}rentiels qui commutent ii}, Mathematische Zeitschrift
  \textbf{171} (1980), 51--73.

\bibitem{Cushman2002b}
R.~Cushman and B.~Zhilinskii, \emph{Monodromy of a two degrees of freedom
  {L}iouville integrable system with many focus-focus singular points}, Journal
  of Physics A: Mathematical and General \textbf{35} (2002), no.~28,
  L415--L419.

\bibitem{Cushman2015}
R.~H. Cushman and L.~M. Bates, \emph{Global aspects of classical integrable
  systems}, 2 ed., Birkh{\"a}user, 2015.

\bibitem{Cushman1988}
R.~H. Cushman and J.~J. Duistermaat, \emph{The quantum mechanical spherical
  pendulum}, Bulletin of the American Mathematical Society \textbf{19} (1988),
  no.~2, 475--479.

\bibitem{Cushman2001}
\bysame, \emph{Non-hamiltonian monodromy}, Journal of Differential Equations
  \textbf{172} (2001), no.~1, 42--58.

\bibitem{Cushman1985}
R.~H. Cushman and H.~Kn{\"o}rrer, \emph{The energy momentum mapping of the
  {L}agrange top}, Differential Geometric Methods in Mathematical Physics,
  Lecture Notes in Mathematics, vol. 1139, Springer, 1985, pp.~12--24.

\bibitem{Cushman2000}
R.~H. Cushman and D.~A. Sadovski{\'\i}, \emph{Monodromy in the hydrogen atom in
  crossed fields}, Physica D: Nonlinear Phenomena \textbf{142} (2000), no.~1-2,
  166--196.

\bibitem{Cushman2002}
R.~H. Cushman and S.~{V{\~u} Ng{\d o}c}, \emph{Sign of the monodromy for
  {L}iouville integrable systems}, Annales Henri Poincar\'{e} \textbf{3}
  (2002), no.~5, 883--894.

\bibitem{Delos2008}
J.~B. Delos, G.~Dhont, D.~A. Sadovski{\'{\i}}, and B.~I. Zhilinski{\'{\i}},
  \emph{Dynamical manifestation of {H}amiltonian monodromy}, {EPL} (Europhysics
  Letters) \textbf{83} (2008), no.~2, 24003.

\bibitem{Derezinski2013}
J.~Derezinski and C.~Gerard, \emph{Scattering theory of classical and quantum
  n-particle systems}, Theoretical and Mathematical Physics, Springer Berlin
  Heidelberg, 2013.

\bibitem{Duistermaat1980}
J.~J. Duistermaat, \emph{On global action-angle coordinates}, Communications on
  Pure and Applied Mathematics \textbf{33} (1980), no.~6, 687--706.

\bibitem{Duistermaat1998}
\bysame, \emph{The monodromy in the {H}amiltonian {H}opf bifurcation},
  Zeitschrift f\"{u}r Angewandte Mathematik und Physik (ZAMP) \textbf{49}
  (1998), no.~1, 156.

\bibitem{DullinPelayo2016}
H.~R. Dullin and {\'A}.~Pelayo, \emph{Generating hyperbolic singularities in
  semitoric systems via {H}opf bifurcations}, Journal of Nonlinear Science
  \textbf{26} (2016), no.~3, 787--811.

\bibitem{DullinWaalkens2008}
H.~R. Dullin and H.~Waalkens, \emph{Nonuniqueness of the phase shift in central
  scattering due to monodromy}, Phys. Rev. Lett. \textbf{101} (2008), 070405.

\bibitem{Efstathiou2005}
K.~Efstathiou, \emph{Metamorphoses of {H}amiltonian systems with symmetries},
  Springer, Berlin Heidelberg New York, 2005.

\bibitem{Efstathiou2013}
K.~Efstathiou and H.~W. Broer, \emph{Uncovering fractional monodromy},
  Communications in Mathematical Physics \textbf{324} (2013), no.~2, 549--588.

\bibitem{Efstathiou2007}
K.~Efstathiou, R.H. Cushman, and D.A. Sadovskii, \emph{Fractional monodromy in
  the 1:-2 resonance}, Advances in Mathematics \textbf{209} (2007), no.~1,
  241--273.

\bibitem{Efstathiou2017}
K.~Efstathiou, A.~Giacobbe, P.~Marde{\v s}i{\'c}, and D.~Sugny, \emph{Rotation
  forms and local hamiltonian monodromy}, Journal of Mathematical Physics
  \textbf{58} (2017), no.~2, 022902.

\bibitem{EfstathiouMartynchuk2017}
K.~Efstathiou and N.~Martynchuk, \emph{Monodromy of {H}amiltonian systems with
  complexity-1 torus actions}, Geometry and Physics \textbf{115} (2017),
  104--115.

\bibitem{Flaschka1988}
H.~Flaschka, \emph{A remark on integrable {H}amiltonian systems}, Physics
  Letters A \textbf{131} (1988), no.~9, 505 -- 508.

\bibitem{Fomenko1990}
A.~T. Fomenko and H.~Zieschang, \emph{{Topological invariant and a criterion
  for equivalence of integrable {H}amiltonian systems with two degrees of
  freedom}}, {Izv. Akad. Nauk SSSR, Ser. Mat.} \textbf{54} (1990), no.~3,
  546--575 (Russian).

\bibitem{Giacobbe2008}
A.~Giacobbe, \emph{Fractional monodromy: parallel transport of homology
  cycles}, Differential Geometry and its Applications \textbf{26} (2008),
  no.~2, 140--150.

\bibitem{Guillemin1989}
V.~Guillemin and A.~Uribe, \emph{Monodromy in the quantum spherical pendulum},
  Communications in Mathematical Physics \textbf{122} (1989), 563?574.

\bibitem{Izosimov2011}
A.M. Izosimov, \emph{Smooth invariants of focus-focus singularities}, Moscow
  Univ. Math. Bull. \textbf{66} (2011), 178.

\bibitem{JaynesCummings1963}
E.~T. Jaynes and F.W. Cummings, \emph{Comparison of quantum and semiclassical
  radiation theories with application to the beam maser}, Proceedings of the
  IEEE \textbf{51} (1963), no.~1, 89--109.

\bibitem{Jung1993}
C.~Jung, \emph{Connection between conserved quantities of the {H}amiltonian and
  of the {S}-matrix}, Journal of Physics A: Mathematical and General
  \textbf{26} (1993), no.~5, 1091.

\bibitem{Knauf1999}
A.~Knauf, \emph{Qualitative aspects of classical potential scattering}, Regul.
  Chaotic Dyn. \textbf{4} (1999), no.~1, 3--22.

\bibitem{Knauf2011}
\bysame, \emph{Mathematische {P}hysik}, Springer-Lehrbuch Masterclass, Springer
  Berlin Heidelberg, 2011.

\bibitem{Kolmogorov1954}
A.~N. Kolmogorov, \emph{Preservation of conditionally periodic movements with
  small change in the {H}amilton function}, Dokl. Akad. Nauk. SSSR \textbf{98}
  (1954), 527.

\bibitem{Kozlov1979}
V.V. Kozlov, \emph{Topological obstructions to the integrability of natural
  mechanical systems}, Soviet Math. Dokl. \textbf{20} (1979), 1413--1415.

\bibitem{Kozlov1983}
\bysame, \emph{Integrability and non-integrability in {H}amiltonian mechanics},
  Russian Math. Surveys \textbf{38} (1983), 1--76.

\bibitem{Lerman1994}
L.~M. Lerman and Ya.~L. Umanski{\u \i}, \emph{Classification of
  four-dimensional integrable {H}amiltonian systems and {P}oisson actions of
  {$\mathbb{R}^2$} in extended neighborhoods of simple singular points. {I}},
  Russian Academy of Sciences. Sbornik Mathematics \textbf{77} (1994), no.~2,
  511--542.

\bibitem{Leung2003}
N.C. Leung and M.~Symington, \emph{Almost toric symplectic four-manifolds},
  Journal of Symplectic Geometry \textbf{8} (2010), no.~2, 143--187.

\bibitem{Liouville1855}
J.~Liouville, \emph{Note sur l'int\'{e}gration des \'{e}quations
  diff\'{e}rentielles de la dynamique, pr\'{e}sent\'{e}e au {B}ureau des
  {L}ongitudes le 29 juin 1853.}, Journal de math{\'e}matiques pures et
  appliqu{\'e}es \textbf{20} (1855), 137--138.

\bibitem{Martynchuk2018}
N.~Martynchuk, \emph{On monodromy in integrable {H}amiltonian systems}, Ph.D.
  thesis, University of Groningen, 2018.

\bibitem{Martynchuk2020}
N.~Martynchuk, H.~W. Broer, and K.~Efstathiou, \emph{Hamiltonian monodromy and
  {M}orse theory}, Communications in Mathematical Physics (2019).

\bibitem{Martynchuk2019}
N.~Martynchuk, H.R. Dullin, K.~Efstathiou, and H.~Waalkens, \emph{Scattering
  invariants in {E}uler's two-center problem}, Nonlinearity \textbf{32} (2019),
  no.~4, 1296--1326.

\bibitem{Martynchuk2017}
N.~Martynchuk and K.~Efstathiou, \emph{Parallel transport along {S}eifert
  manifolds and fractional monodromy}, Communications in Mathematical Physics
  \textbf{356} (2017), no.~2, 427--449.

\bibitem{Martynchuk2016}
N.~Martynchuk and H.~Waalkens, \emph{Knauf's degree and monodromy in planar
  potential scattering}, Regular and Chaotic Dynamics \textbf{21} (2016),
  no.~6, 697--706.

\bibitem{Matsumoto1989}
Y.~Matsumoto, \emph{Topology of torus fibrations}, Sugaku Expositions
  \textbf{2} (1989), 55--73.

\bibitem{Matveev1996}
V.~S. Matveev, \emph{Integrable {H}amiltonian system with two degrees of
  freedom. {T}he topological structure of saturated neighbourhoods of points of
  focus-focus and saddle-saddle type}, Sbornik: Mathematics \textbf{187}
  (1996), no.~4, 495--524.

\bibitem{Meesters2017}
S.F.S Meesters, \emph{Monodromy in the unbounded two-center problem}, BSc
  Thesis, University of Groningen, 2017.

\bibitem{Moser1967}
J.~Moser, \emph{Convergent series expansions for quasi-periodic motions},
  Mathematische Annalen \textbf{169} (1967), no.~1, 136--176.

\bibitem{Nekhoroshev1972}
N.~N. Nekhoroshev, \emph{Action-angle variables, and their generalizations},
  Trans. Moscow Math. Soc. \textbf{26} (1972), 181--198.

\bibitem{Nekhoroshev2006}
N.N. Nekhoroshev, D.A. Sadovski\'{i}, and B.I. Zhilinski\'{i}, \emph{Fractional
  {H}amiltonian monodromy}, Annales Henri Poincar\'{e} \textbf{7} (2006),
  1099--1211.

\bibitem{Pelayo2012}
A.~Pelayo and S.~{V{\~u} Ng{\d o}c}, \emph{Hamiltonian {D}ynamical and
  {S}pectral {T}heory for {S}pin-oscillators}, Communications in Mathematical
  Physics \textbf{309} (2012), no.~1, 123--154.

\bibitem{Poschel1982}
J.~P{\"o}schel, \emph{Integrability of {H}amiltonian systems on {C}antor sets},
  Communications on Pure and Applied Mathematics \textbf{35} (1982), no.~5,
  653--696.

\bibitem{Rink2004}
B.~W. Rink, \emph{A {C}antor set of tori with monodromy near a focus--focus
  singularity}, Nonlinearity \textbf{17} (2004), no.~1, 347--356.

\bibitem{Sadovski2016}
D.A. Sadovsk{\'\i}, \emph{{N}ekhoroshev's approach to {H}amiltonian monodromy},
  Regular and Chaotic Dynamics \textbf{21} (2016), no.~6, 720--758.

\bibitem{Sadovskii1999}
D.~A. Sadovski{\'\i} and B.~I. Zhilinski{\'\i}, \emph{Monodromy, diabolic
  points, and angular momentum coupling}, Physics Letters A \textbf{256}
  (1999), no.~4, 235--244.

\bibitem{Schmidt2010}
S.~Schmidt and H.~R. Dullin, \emph{Dynamics near the $p:-q$ resonance}, Physica
  D: Nonlinear Phenomena \textbf{239} (2010), no.~19, 1884--1891.

\bibitem{Sepe2018}
V.~Sepe and S.~V{\~u}~Ng{\d o}c, \emph{Integrable systems, symmetries, and
  quantization}, Lett. Math. Phys. \textbf{108} (2018), 499--571.

\bibitem{Smirnov2013}
G.E. Smirnov, \emph{Focus-focus singularities in classical mechanics}, Rus. J.
  Nonlin. Dyn. \textbf{10} (2014), no.~1, 101--112.

\bibitem{Sugny2008}
D.~Sugny, P.~Marde{\v s}i{\'c}, M.~Pelletier, A.~Jebrane, and H.~R. Jauslin,
  \emph{Fractional hamiltonian monodromy from a {G}auss--{M}anin monodromy},
  Journal of Mathematical Physics \textbf{49} (2008), no.~4, 042701.

\bibitem{Takens2010}
F.~Takens, \emph{Private communication}, 2010.

\bibitem{Tarama2012}
D.~Tarama, \emph{Elliptic {K}3 surfaces as dynamical models and their
  hamiltonian monodromy}, Central European Journal of Mathematics \textbf{10}
  (2012), 1619--1626.

\bibitem{Tonkonog2013}
D.I. Tonkonog, \emph{A simple proof of the geometric fractional monodromy
  theorem}, Moscow University Mathematics Bulletin \textbf{68} (2013), no.~2,
  118--121.

\bibitem{Vu-Ngoc1999}
S.~{V{\~u} Ng{\d o}c}, \emph{Quantum monodromy in integrable systems},
  Communications in Mathematical Physics \textbf{203} (1999), no.~2, 465--479.

\bibitem{Vu-Ngoc2000}
\bysame, \emph{{B}ohr-{S}ommerfeld conditions for integrable systems with
  critical manifolds of focus-focus type}, Communications on Pure and Applied
  Mathematics \textbf{53} (2000), no.~2, 143--217.

\bibitem{Vu-Ngoc2003}
S.~V{\~u}~Ng{\d o}c, \emph{On semi-global invariants for focus-focus
  singularities}, Topology \textbf{42} (2003), 365--380.

\bibitem{Waalkens2004}
H.~Waalkens, H.~R. Dullin, and P.~H. Richter, \emph{The problem of two fixed
  centers: bifurcations, actions, monodromy}, Physica D: Nonlinear Phenomena
  \textbf{196} (2004), no.~3-4, 265--310.

\bibitem{Whitney1934}
H.~Whitney, \emph{Analytic extensions of differentiable functions defined in
  closed sets}, Transactions of the American Mathematical Society \textbf{36}
  (1934), no.~1, 63--89.

\bibitem{Zhilinskii2005}
B.~I. Zhilinski{\'\i}, \emph{Interpretation of quantum hamiltonian monodromy in
  terms of lattice defects}, Acta Applicandae Mathematicae \textbf{87} (2005),
  no.~1-3, 281--307.

\bibitem{Zhilinskii2010}
Boris Zhilinskii, \emph{Quantum monodromy and pattern formation}, Journal of
  Physics A: Mathematical and Theoretical \textbf{43} (2010), no.~43, 434033.

\bibitem{Zung1997}
N.~T. Zung, \emph{A note on focus-focus singularities}, Differential Geometry
  and its Applications \textbf{7} (1997), no.~2, 123--130.

\bibitem{Zung2002}
\bysame, \emph{Another note on focus-focus singularities}, Letters in
  Mathematical Physics \textbf{60} (2002), no.~1, 87--99.

\bibitem{Zung1996}
N.T. Zung, \emph{Kolmogorov condition for integrable systems with focus-focus
  singularities}, Physics Letters A \textbf{215} (1996), no.~1, 40--44.

\end{thebibliography}

\end{document}